\documentclass[authoryear,review]{elsarticle}
\makeatletter
\def\ps@pprintTitle{%
	\let\@oddhead\@empty
	\let\@evenhead\@empty
	\def\@oddfoot{}%
	\let\@evenfoot\@oddfoot}
\makeatother

\usepackage[utf8]{inputenc}
\usepackage{fontenc}[t1]
\usepackage[hidelinks]{hyperref}
\usepackage{mathptmx} 
\usepackage{setspace}
\usepackage{amssymb}
\usepackage{graphicx}
\usepackage{subfig}
\usepackage{caption}
\usepackage{amsfonts}
\usepackage{hyperref}
\usepackage{amsmath}
\usepackage{amsthm} 
\theoremstyle{plain}

\usepackage{color}
\usepackage{colordvi}

\newtheorem{prop}{Proposition}[section]
\newtheorem{defi}[prop]{Definition}
\newtheorem{rem}[prop]{Remark}

\newtheorem{theorem}[prop]{Theorem}

\newtheorem{cor}[prop]{Corollary}

\newcommand{\komment}[1]{}

\usepackage[right]{lineno}

\begin{document}

\begin{frontmatter}
	
	\title{Exact and approximate formulas for contact tracing on random trees}
	
	\author[TUMGarching]{Augustine Okolie\corref{mycorrespondingauthor}} 
	\author[TUMGarching,ICB]{Johannes M\"uller} 
	
	\address[TUMGarching]{Center for Mathematical Sciences, Technische Universit\"at M\"unchen, 85748 Garching, Germany}
	\address[ICB]{Institute for Computational Biology, Helmholtz Center Munich, 85764 Neuherberg, Germany
	\par\medskip
	Received at Mathematical Biosciences 16 October 2019; Received in revised form 23 January 2020;\\ Accepted 23 January 2020; Available online 31 January 2020\\
	DOI \url{10.1016/j.mbs.2020.108320}}

\cortext[mycorrespondingauthor]{Corresponding author.\\
 Email address: \url{augustine.okolie@tum.de} (A.Okolie).}


\begin{abstract}
We consider a stochastic susceptible-infected-recovered (SIR) model with contact tracing on random trees and on the configuration model. On a rooted tree, where initially all individuals are susceptible apart from the root which is infected, we are able to find exact formulas for the distribution of the infectious period. Thereto, we show how to extend the existing theory for contact tracing in homogeneously mixing populations to trees. Based on these formulas, we discuss the influence of randomness in the tree and the basic reproduction number. We find the well known results for the homogeneously mixing case as a limit of the present model (tree-shaped contact graph). Furthermore, we develop approximate mean field equations for the dynamics on trees, and -- using the message passing method -- also for the configuration model. The interpretation and implications of the results are discussed.
\end{abstract}
	
\end{frontmatter}


\section{Introduction}

Most control measures, such as mass screening or mass vaccination, involve testing and treating random samples of individuals in a population \citep{bauch2000moment, ferguson2000more}. In contrast, contact tracing or ring vaccination works by further targetting partners or neighbours of an identified infected individual (called an 'index case'). For contact tracing to take effect, network information about local contact structures and correlations among individuals have to be accounted for \citep{klinkenberg2006effectiveness}. Model analysis showed that contact tracing reduces disease prevalence by breaking chains associated with routes of transmission \citep{eames2007contact}, a point that we will discuss below from a different perspective. Several intensive simulation studies have investigated the influence of contact graph on contact tracing~\citep{Kiss2005,Kiss2007}. From an ecological point of view, contact tracing has been dubbed "hyperparasitism" because it spreads through the network in the same way as the disease but targets only infected individuals
 \citep{mills1996prospective, sullivan1999hyperparasitism}.
\par\medskip

The analysis of models for contact tracing is more involving and mathematically challenging than models e.g.\ targeting random screening, as it is necessary to address local correlations. 
Several approaches are present in the literature. One major approach that focuses on correlations is pair approximation~\citep{keeling1999,house2010,eames2003contact}. Originally, pair approximation has been developed to better analyse the effect of correlations in stochastic processes on graphs, e.g., infected individuals tend to cluster, and therefore the spread of an epidemic slows down. As contact tracing targets on links between infected individuals, pair approximation is useful also in that context. Frazer et al.~\citep{fraser2004} did focus on the age of infection and analysed an (deterministic) age-structured model. In an homogeneously mixing population, an approach often used is the analysis of a branching process with dependencies~\citep{muller2000contact,klinkenberg2006effectiveness,muller2007estimating,muller2016effect,Ball2011,Ball2015}. The advantage of this method is an exact analysis of the corresponding models, at least during the onset of an epidemic. It is possible to handle an SIS and SIR model, as long as the onset of the infection is considered. Up to now, this method was restricted to a complete contact graph and a large (infinite) population. The analysis of epidemics on more realistic graphs is more challenging ~\citep{barbour1990epidemics}.
\par\medskip 

However, the underlying idea of the branching process approach to contact tracing resembles that of the message passing method for a stochastic SIR dynamics on a tree~\citep{karrer2010message,wilkinson2014message, kiss2015generalization}. Starting from this observation, it became clear that it is possible to extend the branching theory for contact tracing to stochastic SIR models on graphs. And indeed, as we will show below, it is possible to obtain results resembling that in~\citep{muller2000contact,muller2016effect}. The present work starts off with an SIR model on a contact graph that is a tree.  Initially, only the root is infected while all other nodes are susceptible. It is possible to obtain integro-differential equations for the probability to be infectious at a given age of infection. Central information, as the reproduction number, can be analysed, and a suited limit connects the results on a graph with well known results for homogeneous populations ~\citep{muller2000contact,muller2016effect}. From that micro-level analysis, we derive by means of heuristic analysis, a mean field equation. Based on even more heuristic arguments, we apply the theory to a non-tree contact graph: the configuration model.

\section{Model and analysis}

We consider a stochastic SIR model on a contact graph given by a rooted random tree \citep{newman2002random}. The nodes represent individuals, and can assume the states ``susceptible'', ``infected and infectious'', or ``recovered''. At time zero, we assume that the root is the only infected individual while all other individuals are susceptible. Importantly, the root has no infector. Two individuals connected by an edge have contacts at rate $\beta$. That is, the time between two contacts is exponentially distributed with rate $\beta$; each contact between an infected and a susceptible node transmits the infection. 
Individuals recover either undetected at rate $\alpha$, or they become diagnosed (rate $\sigma$). Diagnosed individuals are treated and recover. At the same time, the diagnosed individuals become index cases. Index cases are asked to name their contact partners, who then are called in to see a doctor. In that, an infected neighbour of an index case has probability $p$ to be discovered. In one-step tracing, we stop after one step, while in recursive tracing, infected individuals that are detected by contact tracing form new index cases. Any recovered individual acquires a life-long immunity.\par\medskip

We denote a node (an individual) $A$ as a "downstream" of node $B$, if the shortest path between the root and node $A$ passes through node $B$ (additionally, all non-root individuals are downstream from the root's perspective). Similarly, an edge of a node is called downstream if it connects to a downstream node. The random variable $K$ denotes the number of downstream edges, such that the total number of edges for a non-root individual is $K+1$. The number of downstream nodes are assumed to be i.i.d. for all individuals.\par\medskip 

Below we investigate the probability for an infected individual to still be infectious at age of infection $a$ (often, we will just write ``age'', as we do not use chronological age but only age of infection). This probability depends in general on the location of the individual within the contact tree because of contact tracing. Let the generation $i$ of a node be the distance between the node and the root (w.r.t.\ the standard graph metric). We define 
\begin{eqnarray}
\kappa_i(a) = P(\mbox{a randomly chosen infected node of generation } i \mbox{ is infectious at age of infection }a),
\end{eqnarray}
which turns out to be central for our analysis of the process. 
In order to obtain this probability, we divide the tracing process in ``backward'' and ``forward'' tracing: We call a tracing event ``backward tracing'' if an individual is detected via a downstream index case, and ``forward tracing'', if the index case is upstream. In backward tracing we artificially switch off all forward  tracing events, while in forward tracing, we only allow for forward tracing events. Of course, this is done for mathematical convenience only. In reality, we always find the full tracing process. This full process can be easily understood as a combination of forward- and backward tracing~\citep{muller2000contact}. \par\medskip 
{\bf Notation:} The symbol $\kappa$ will appear with several sub- and superscripts. In order to avoid confusion, we summarize the different roles of $\kappa$ here: Apart from forward- backward- and full tracing, we will consider one step and recursive tracing. Therefore, we use $\kappa_{*,i}^+(a)$ for forward tracing, $\kappa_{*,i}^-(a)$ for backward tracing, and $\kappa_{*,i}(a)$ for full tracing. The asterisk 
$\ast$ is either ``r'' (for recursive) or ``o'' (for one-step tracing). That is, $\kappa_{r,i}^+(a)$ refers to the probability to be infectious at age (of infection) $a$ for an individual of generation $i$, subject to recursive forward tracing. We furthermore denote by  $$\widehat\kappa(a)=e^{-(\alpha+\sigma)a}$$ 
the probability to be infectious at age $a$ if no tracing takes place ($p=0$). Last, in backward tracing, we note that only the downstream individuals of an (infected) focal individual would be ``actors'' in the tracing process. Hence, the generation of the focal individual plays no role such that the contact graph becomes an infinite tree. So $\kappa_\ast^-(a)$ is the (generation-independent) probability in case of backward tracing. Importantly, as we don't know the exact location of this randomly selected focal individual with $k$ downstream edges on the contact graph, we define $\kappa_{*,i+1}^-(\cdot)$ to be the expected value of $\kappa_{*,k,i+1}^-(\cdot)$ (the probabilities for each of the $k$ downstream individuals of a randomly selected infected individual to still be infectious at some age $\cdot$) over all possible $k$, where the random variable $K$ takes on $k$.

\subsection{Backward tracing}
In backward tracing, an infected individual can only be traced through his/her infectee but not through his/her infector. In the following subsections, we will discuss the recursive and one-step mode of backward tracing.

\subsubsection{Recursive contact tracing}
Let  $\kappa_{r,i}^-(a)$ denote the probability that a focal individual of  generation $i$ is still infectious at age $a$ of infection, if only recursive backward tracing takes place (no forward tracing).

\begin{theorem}\label{backRecurs}
Let $G(s) = E(s^{K})$. We find  $\kappa_{r,i}^-(a) = \kappa_r^-(a)$, where $\kappa_r^-(a)$ is determined by
\begin{eqnarray} \label{eq:2}
\kappa_r^-(a) 
&=& 
e^{-(\alpha+\sigma)a}\, 
G\left(\,\,
1- p\int_0^a\bigg[1-e^{-\beta (a-\tilde a)}\, \bigg]
\bigg[-\frac d {d \tilde a}\kappa_r^-(\tilde a)\, 
-\alpha\kappa_r^-(\tilde a)\,\bigg]\, d\tilde a\,\, \right).
\end{eqnarray}
\end{theorem}

\begin{proof}(of Theorem~\ref{backRecurs}) 
Let $\kappa_{r,k,i}^-(a)$ denote the probability of an individual in the $i$'th generation with $i>0$, that has $k$ downstream edges, to be infectious at age $a$ of infection. An individual recovers spontaneously at rate $\alpha$, and is directly diagnosed at rate $\sigma$. Hence, 
\begin{equation} \label{eq:5}
\kappa_{r,k,i}^ - (a) = e^{-(\alpha+\sigma)a}\,\,[1 - P({\text{tracing event in the interval [0,a]}})].
\end{equation}

The probability of no tracing event in the interval $[0, a]$ could for all we know at this point, depend on $k$ and $i$.

In order to determine $P({\text{tracing event in the interval [0,a]}})$, we first focus on a single downstream edge. The contactee at this downstream edge has contacts with the (infected) focal individual at rate $\beta$. The probability density that the first contact on this edge (after the focal individual became infected) took place at age $\tilde a$ of infection (of the focal individual)  is given by
$$ \beta \, e^{-\beta \tilde a}.$$
Note that the contactee becomes infected during this first contact. 
Furthermore, the hazard of a given downstream individual infected at age $\tilde a$ of the focal individual at age $b>\tilde a$ (again, of the focal individual)  reads
$$ \frac{-\frac d {d b}\kappa_{r,i+1}^-(b-\tilde a)}
        { \kappa_{r,i+1}^-(b-\tilde a)}.$$
However, we do not need the hazard, but the rate at which the individual is detected, no matter if directly or
per tracing. Therefore we subtract from the hazard the rate for unobserved recovery, 
\begin{eqnarray}\label{backTracinfgRate}
 \frac{-\frac d {d b}\kappa_{r,i+1}^-(b-\tilde a)}
{ \kappa_{r,i+1}^-(b-\tilde a)}-\alpha.
\end{eqnarray}
Hence, the {\it rate} for the focal individual at age $b$ (where $b<a$) to be traced via a given edge reads (in the following equation, 
$b$ is always some age of the focal individual)
\begin{eqnarray*}
&& p\,\int_0^b 
(\mbox{downstream indiv.\ becomes infected at }\tilde a)\,
\times (\mbox{downstream indiv.\ still infectious at age }b)\,\\\
&&\qquad\qquad\times (\mbox{detection rate of downstream indiv.\ at age }b)\, d\tilde a\\
&=&p\, \int_0^b\,  \bigg(\beta \, e^{-\beta \tilde a}\bigg)\, 
\kappa_{r,i+1}^-(b-\tilde a)\, 
\left( \frac{-\frac d {d a}\kappa_{r,i+1}^-(b-\tilde a)}
{ \kappa_{r,i+1}^-(b-\tilde a)}-\alpha\right)\, d\tilde a
=p\, \int_0^b\,  \beta \, e^{-\beta (b-\tilde a)}\, 
\bigg(-\frac d {d \tilde a}\kappa_{r,i+1}^-(\tilde a)\, 
-\alpha\kappa_{r,i+1}^-(\tilde a)\,\bigg)\, d\tilde a.
\end{eqnarray*}
We need to integrate over $b$ to obtain the probability for the focal individual that a tracing event took place via the given edge before age $a$. 
The probability to not be traced via a given edge reads
$$1-p\, \int_0^a\int_0^b\,  \beta \, e^{-\beta (b-\tilde a)}\, 
\bigg(-\frac d {d \tilde a}\kappa_{r,i+1}^-(\tilde a)\, 
-\alpha\kappa_{r,i+1}^-(\tilde a)\,\bigg)\, d\tilde a\, db.$$
As we have $k$ edges, 
$$\kappa_{r,k,i}(a) 
= \widehat\kappa(a)\, 
\left(1-p\, \int_0^a\int_0^b\,  \beta \, e^{-\beta (b-\tilde a)}\, 
\bigg(-\frac d {d \tilde a}\kappa_{r,i+1}^-(\tilde a)\, 
-\alpha\kappa_{r,i+1}^-(\tilde a)\,\bigg)\, d\tilde a\, db\right)^k.$$
Now comes an important ingredient: As we consider backward tracing, only downstream nodes can trigger tracing events. Therefore, the generation of an individual does not affect  the probability to be infections at a given age, 
$$\kappa_{r,i}^-(a)=\kappa_{r,j}^-(a)\qquad \mbox{for }i,j\geq 0.$$
That is, we may write $\kappa_{r,i}^-(a)=\kappa_r^-(a)$. Last, 
we remove the condition $K=k$ for the focal individual, 
\begin{eqnarray*}
\kappa_r^-(a) & = & 
\sum_{k=0}^\infty \kappa_{r,k,i}\,P(K=k) 
= 
 \widehat\kappa(a)\,\sum_{k=0}^\infty\, 
\left( 1-p\, \int_0^a\int_0^b\,  \beta \, e^{-\beta (b-\tilde a)}\, 
\bigg(-\frac d {d \tilde a}\kappa_{r,i+1}^-(\tilde a)\, 
-\alpha\kappa_{r,i+1}^-(\tilde a)\,\bigg)\, d\tilde a\, db\right)^k\,\,P(K=k) \\
&= &
\widehat\kappa(a)\, G\left( 1-p\, \int_0^a\int_0^b\,  \beta \, e^{-\beta (b-\tilde a)}\, 
\bigg(-\frac d {d \tilde a}\kappa_r^-(\tilde a)\, 
-\alpha\kappa_r^-(\tilde a)\,\bigg)\, d\tilde a\, db\right).
\end{eqnarray*}
This equation is basically our result. At the end of the proof, we only rewrite the integral, 
\begin{eqnarray*}
 \int_0^a\int_0^b\,  \beta \, e^{-\beta (b-\tilde a)}\, 
\bigg(-\frac d {d \tilde a}\kappa_r^-(\tilde a)\, 
-\alpha\kappa_r^-(\tilde a)\,\bigg)\, d\tilde a\, db
&=&  \int_0^a\left(\int_{\tilde a}^a\,  -1\,\frac d {db} e^{-\beta (b-\tilde a)}\, db\,\right)\,\,
\bigg(-\frac d {d \tilde a}\kappa_r^-(\tilde a)\, 
-\alpha\kappa_r^-(\tilde a)\,\bigg)\, d\tilde a\\
&=&  \int_0^a\bigg(1-e^{-\beta (a-\tilde a)}\, \bigg)
\bigg(-\frac d {d \tilde a}\kappa_r^-(\tilde a)\, 
-\alpha\kappa_r^-(\tilde a)\,\bigg)\, d\tilde a\,.
\end{eqnarray*}
\end{proof} 

In general, the integral equation (\ref{eq:2}) will not have an explicit solution. However, we can use Taylor expansion to obtain a first order approximation of $\kappa_r^-(a)$ in $p$. 

\begin{prop}Let $\beta  \ne \alpha  + \sigma $, and assume that all moments of $K$ are finite.  The first order approximation of $\kappa _r^ - (a)$ reads
\begin{equation} \label{eq:17}
\kappa _r^ - (a) = \hat \kappa (a)\left( {1 + \frac{{p\sigma E(K)}}{{\alpha  + \sigma  - \beta }}\left( {{e^{ - \beta a}} - \hat \kappa (a)} \right) - \frac{{p\sigma E(K)}}{{\alpha  + \sigma }}\left( {1 - \hat \kappa (a)} \right)} \right) + \mathcal{O} ({p^2}).
\end{equation}
\end{prop}
\begin{proof}
First of all, $\kappa_r^-(a)|_{p=0}=\widehat\kappa(a)$. As we aim at a first order approximation, we are allowed to replace all terms of the form $p\,\kappa_r^-(a)$ by 
$p\, \widehat\kappa(a)=p\, e^{-(\alpha+\sigma)\,a}$. Therewith, eqn.\ (\ref{eq:2}) 
becomes
\begin{eqnarray*}
\kappa_r^-(a) 
&=& 
e^{-(\alpha+\sigma)a}\, 
G\left(\,\,
1- p\int_0^a\bigg[1-e^{-\beta (a-\tilde a)}\, \bigg]
\bigg[-\frac d {d \tilde a} e^{-(\alpha+\sigma)\,\tilde a}\, 
-\alpha e^{-(\alpha+\sigma)\,\tilde a}\,\bigg]\, d\tilde a\,\, \right)+{\cal O}(p^2)\\
&=& 
e^{-(\alpha+\sigma)a}\, 
G\left(\,\,
1- p\,\sigma\int_0^a\bigg[1-e^{-\beta (a-\tilde a)}\, \bigg]
 e^{-(\alpha+\sigma)\,\tilde a}\, d\tilde a\,\, \right)+{\cal O}(p^2)\\
&=& 
e^{-(\alpha+\sigma)a}\, 
G\left(\,\,
1- p\,\bigg[
\frac{\sigma}{\alpha+\sigma}(1-e^{-(\alpha+\sigma)a}) 
- 
\frac{\sigma}{\alpha+\sigma-\beta}\,e^{-\beta a}\,
       (1-e^{-(\alpha+\sigma-\beta)a}) 
\bigg]\,\, \right)+{\cal O}(p^2)\\
&=& 
e^{-(\alpha+\sigma)a}\, 
\bigg(G(1)-p\, G'(1)\, 
\,\bigg[
\frac{\sigma}{\alpha+\sigma}(1-\widehat\kappa(a)) 
- 
\frac{\sigma}{\alpha+\sigma-\beta}(e^{-\beta a}-\widehat\kappa(a)) 
\bigg]\,\,
\bigg)+{\cal O}(p^2).
\end{eqnarray*}
As $G(1)=1$ and $G'(1)=E(K)$, this equation establishes the result.
\end{proof}

As illustrated in Fig.~\ref{fig:BackwardRecursive}, the tracing term in this result reduces the probability to be infectious. If we compare the numerical solution of eqn.~(\ref{eq:2}) and the approximation given in eqn.~(\ref{eq:17}), we find for the parameters chosen ($p=0.3$), a good agreement with simulation results. 
Find some notes related to the simulation algorithm in the appendix.

\begin{figure}[h!]
	\centering
		\includegraphics[width=\textwidth]{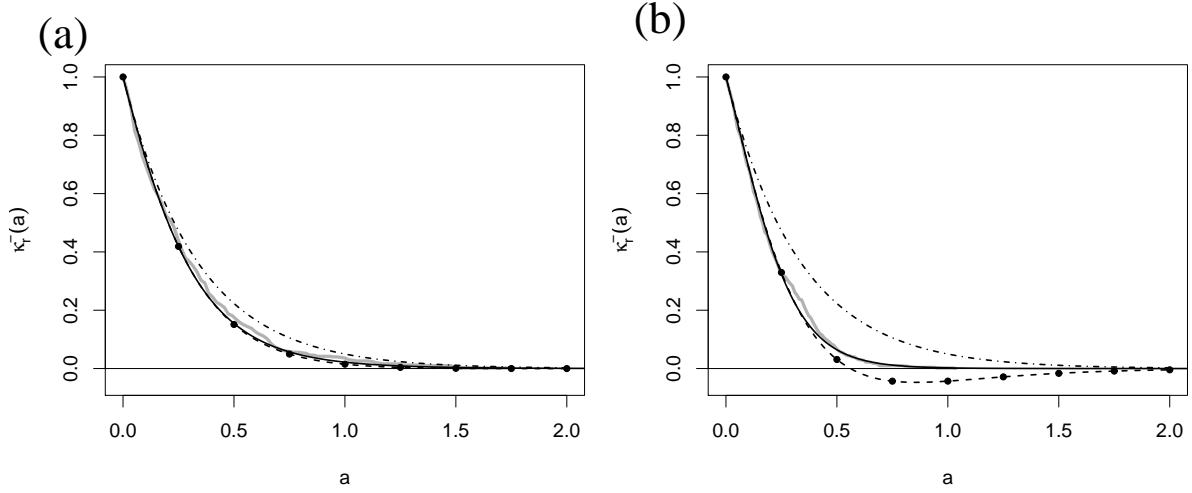}
	\caption{Recursive backward tracing. The probability to be infectious at age $a$ after infection ($\kappa_r^-(a)$) for $p=0.3$ (panel a), and $p=0.8$ (panel b).  Dashed-dotted  line: without tracing
	($\hat \kappa (a)=e^{-(\sigma+\alpha)a}$); solid line: numerical solution of $\kappa_r^-(a)$; dashed-lines with bullets: first order approximation of $\kappa_r^-(a)$; grey line: simulation. 
 Parameters: $\beta=1.5$, $\alpha = 0.1$, $\sigma=2.9$, $E(K)=4$, fixed degree.}
	\label{fig:BackwardRecursive}
\end{figure}

\subsubsection{One-step tracing}

The only difference between the recursive and the one-step method is that in the one-step method, the tracing is aborted after one step. Illustratively, this means that person $A$ can only be discovered through tracing events, which are triggered by neighbours of $A$ but not by tracing events from which neighbours of $A$ are discovered.

\begin{theorem}\label{backOneStep}
The probability that an infected individual under one step backward tracing reaches the age of infection a reads
\begin{equation} \label{eq:28}
\kappa _o^ - (a) = {e^{ - (\alpha  + \sigma )a}}\,G\left( {1 - p\int\limits_0^a {(1 - } {e^{ - \beta (a  - \widetilde a)}})\sigma \kappa _o^ - (\widetilde a)d\widetilde a} \right).
\end{equation}
\end{theorem}
\begin{proof}
Following the same arguments in the recursive mode, the prove is similar to proposition~\ref{backRecurs}. The only difference is the detection rate, given in eqn.~(\ref{backTracinfgRate}). For the present case, this rate simply reads $\sigma$. 
\end{proof}

\begin{prop}
The first order approximation of $\kappa_o^{-}$ reads:
\begin{align}
\kappa _o^ - (a) = \hat \kappa (a)\left( {1 + \frac{{p\sigma E(K)}}{{\alpha  + \sigma  - \beta }}\left( {{e^{ - \beta a}} - \hat \kappa (a)} \right) - \frac{{p\sigma E(K)}}{{\alpha  + \sigma }}\left( {1 - \hat \kappa (a)} \right)} \right) + \mathcal{O} ({p^2}).\label{eq:32}
\end{align}\label{backOneStepApprox}
\end{prop}
\begin{proof}
The proof follows as proposition 2.2. Hence, we can proceed to determine the approximation exactly as with the recursive method.
\end{proof}
Note that the first approximation for one-step tracing coincides with that for recursive tracing. The reason is that tracing one edge has probability ${\cal O}(p)$, such that tracing a path of length 2 already has a higher order term ${\cal O}(p^2)$. Again, 
the theoretical and simulation results fit well as shown in Fig.~\ref{fig:BackwardOneStep}.

\begin{figure}[h!]
	\centering
	\includegraphics[width=\textwidth]{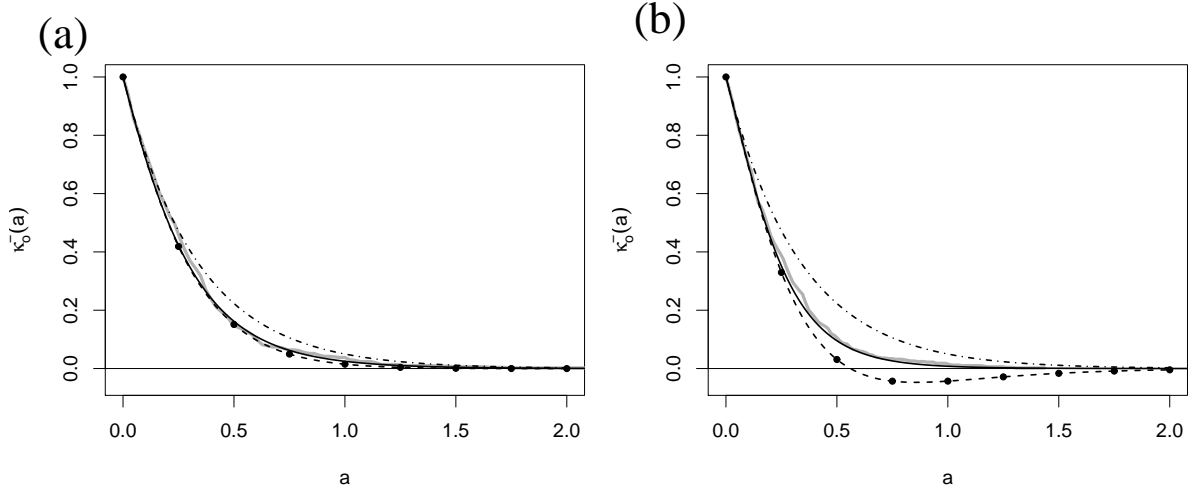}
	\caption{One-step backward tracing. The probability to be infectious at age $a$ after infection ($\kappa_o^-(a)$) for $p=0.3$ (panel a), and $p=0.8$ (panel b).  Dashed-dotted  line: without tracing
	($\hat \kappa (a)=e^{-(\sigma+\alpha)a}$); solid line: numerical solution of $\kappa_o^-(a)$; dashed-lines with bullets: first order approximation of $\kappa_o^-(a)$; grey line: simulation.
		Parameters: $\beta=1.5$, $\alpha = 0.1$, $\sigma=2.9$, $E(K)=4$, fixed degree.}
	\label{fig:BackwardOneStep}
\end{figure}

\subsubsection{Connection to the results for homogeneously mixing populations}
Interestingly enough, we can relate the result for random trees to the corresponding results on randomly mixing population, derived in~\citep{muller2000contact,muller2016effect}. In a randomly mixing population, the contact rate is not defined per edge ($\beta$), but for an individual  ($\beta_{ind}$). The model and the meaning of all other parameters parallel that of the present work, with the exception that the contact graph is not a tree but a complete graph. The corresponding probability $\kappa_{h}^-(a)$ to be infectious at age $a$ of infection for an homogeneous population (considered in the onset of an infection, and one step backward tracing only) is given by the equation 
(Proposition 2.1 and Proposition 2.2 in~\citet{muller2016effect})
\begin{eqnarray} \label{BackHomPopExactModel}
\frac{d}{{da}}\kappa_h^- (a) =  - \kappa _h^ - (a)\left[ {\alpha  + \sigma  + p\,\sigma\, \beta_{ind} \int\limits_0^a {\kappa_h^ - (\tilde a)} d\tilde a} \right],\qquad \kappa_h^-(0)=1,
\end{eqnarray}
while the first order approximation reads
\begin{eqnarray}\label{BackHomPopExactFirstOrder}
{\kappa _h^-}(a) ={}& \hat \kappa (a) - p\frac{\sigma }{{\alpha  + \sigma }}\beta_{ind} \hat \kappa (a)\left( {a - \frac{{1 - \hat \kappa (a)}}{{\alpha  + \sigma }}} \right) + \mathcal{O} ({p^2}).
\end{eqnarray}
One central difference between this present equations and the homogeneous models is the appearance of $\beta_{ind}$. 
If we select randomly a non-root node in the present model, then the total contact rate of that node is the sum of all contacts on its edges. For all individuals apart of the root, the relation between $\beta$ and $\beta_{ind}$ is given by
\begin{eqnarray}
\beta_{ind} = (1+E(K))\, \beta.
\end{eqnarray}
Here we obtain $1+E(K)$, as contacts happen on upstream and on downstream nodes. It is near at hand to consider the limit $E(K)\rightarrow\infty$, $\beta\rightarrow0$, while $\beta_{ind} = (1+E(K))\, \beta$ is constant to approximate a randomly mixing population (full graph). We first discuss the first order approximation.

\begin{prop}\label{approxBackwardHomo1} We find
\begin{equation} \label{eq:35}
\kappa_h^{-}(a) = \lim_{\beta\rightarrow 0,\,E(K)\rightarrow\infty}\kappa_o^-(a) + \mathcal{O}({p^2}),
\end{equation}
under the condition that $\beta_{ind} : = \mathop {\lim }\limits_{\beta  \to 0,E(K) \to \infty }  \beta (1+E(K)).$
\end{prop}
\begin{proof}
First note, $\beta_{ind} : = \mathop {\lim }\limits_{\beta  \to 0,E(K) \to \infty } \beta E(K)$. We take the limit $\beta\rightarrow 0$, $E(K)\rightarrow\infty$ of eqn.~\ref{eq:32}, 
\begin{equation} \label{eq:36}
\kappa _o^ - (a) = \hat \kappa (a)\left( {1 + \frac{{p\sigma E(K)}}{{\alpha  + \sigma  - \beta }}\left( {{e^{ - \beta a}} - \hat \kappa (a)} \right) - \frac{{p\sigma E(K)}}{{\alpha  + \sigma }}\left( {1 - \hat \kappa (a)} \right)} \right) + \mathcal{O} ({p^2}).
\end{equation}
%
We note that
\begin{align*}
{}& {\left. {\mathop {\lim }\limits_{\beta  \to 0} \,\,p\sigma \beta\,E(K)\,\,
\,\, \frac 1 \beta\,\,\bigg\{\frac{1}{{\alpha  + \sigma  - \beta }}\left( {{e^{ - \beta a}} - {e^{ - (\alpha  + \sigma )a}}} \right) - \frac{1}{{\alpha  + \sigma }}\left( {1 - {e^{ - (\alpha  + \sigma )a}}} \right)} \bigg\}\right|_{\beta E(K)\rightarrow\beta_{ind}}},\\
={}& {\left. {\frac{{p\sigma \beta_{ind} }}{{\alpha  + \sigma }}\frac{d}{{d\beta }}\left( {\frac{{\alpha  + \sigma }}{{\alpha  + \sigma  - \beta }}\left( {{e^{ - \beta a}} - {e^{ - (\alpha  + \sigma )a}}} \right)} \right)} \right|_{\beta  = 0}}
= \frac{{p\sigma \beta_{ind} }}{{\alpha  + \sigma }}\left( {\frac{1}{{\alpha  + \sigma }}\left( {1 - {e^{ - (\alpha  + \sigma )a}}} \right) - a} \right).
\end{align*}
Hence, the limit reads
$$
\kappa _o^ - (a) = {e^{ - (\alpha  + \sigma )a}} - \frac{{p\sigma \beta_{ind} }}{{\alpha  + \sigma }}{e^{ - (\alpha  + \sigma )a}}\left( {a - \frac{1}{{\alpha  + \sigma }}\left( {1 - {e^{ - (\alpha  + \sigma )a}}} \right)} \right) + {\cal O}({p^2}).
$$
This equation establishes the result. 
\end{proof}

We have shown above that under the limit for any (well behaved) degree distribution, the first order approximation of our tree model and the random mixing model of \citep{muller2016effect} agree. Now we show an even stronger result. Let us assume that $K$ follows the Poisson process $K \sim Pois\left( {E(K)} \right)$. The random mixing models assumes that contacts happen at randomly distributed times (homogeneous Poisson process) and the number of contactees over a given time span follows a Poisson distribution. Thus, it is natural to examine the tree model with a Poissonian degree distribution. 
\begin{prop}\label{approxBackwardHomo2}
Let $K\sim\mbox{Pois}(E(K))$, and define $\beta_{ind} = (1+E(K))\, \beta$. Eqn.~(\ref{eq:28}) becomes in the limit 
$E(K)\rightarrow\infty$, $\beta\rightarrow0$, while $\beta_ {ind}$ is kept constant, 
\begin{equation}\label{backOneExactHomo}
\frac{d}{{da}}\kappa_o^- a) =  - \kappa _o^ - (a)\left[ {\alpha  + \sigma  + p\sigma \beta_{ind} \int\limits_0^a {\kappa _o^ - (\tilde a)} \, d\tilde a} \right].
\end{equation}
\end{prop}
\begin{proof}
The PGF of a Poissonian distributed random variable reads $G(s)=E({s^K})  = {e^{ - (1 - s)E(K)}}$. With Theorem~\ref{backOneStep} we conclude
\begin{eqnarray*}
\kappa_o^ - (a) = {e^{ - (\alpha  + \sigma )a}}\exp\left( {\,-E(K) \,p\,\sigma\,\int\limits_0^a {(1 - } {e^{ - \beta (a  - \tilde a)}})\,\kappa_o^ - (\tilde a)\,d\tilde a}\,\right).
\end{eqnarray*} 
The derivative w.r.t.\ $a$ gives
\begin{align*}
\frac{d}{{da}}\kappa_o^- (a) =  - \kappa_o^-(a)\left[ {\alpha  + \sigma  + E(K)p\sigma \beta \int\limits_0^a {{e^{ - \beta (a - \tilde a)}}\kappa _o^ - (\tilde a)}\, d\tilde a} \right].
\end{align*}
If we again take  the limit $\beta  \to 0$ while $E(K) \to \infty$, 
while for $\beta (1+E(K))\rightarrow\beta_{ind}$, we obtain eqn.~(\ref{backOneExactHomo}).
\end{proof}
Note that eqn.~(\ref{backOneExactHomo}) is identical with eqn.~(\ref{BackHomPopExactModel}). 
Not only both models agree at the first order approximation in $p$, but also that, in this limit (and that if $K$ follows a Poisson distribution), the models themselves coincide.

\subsection{Forward tracing}
In forward tracing, we note that an infected individual, if not from the zeroth generation, can only be traced through his/her infector. Therefore, in forward tracing, the generation matters. On the other hand, each person only has exactly one infector which is also the case in randomly mixing populations. The line of reasoning is identical with that of homogeneous models~\citep{muller2000contact,muller2016effect}. In particular, the degree distribution does not appear in the theory we develop next. 

\subsubsection{Recursive method}

\begin{defi}
Let $\kappa_{r,i}^+(a|b)$ be the probability that an individual of generation $i$ is still infected at age $a$ of infection given that the infector has age $a + b$ of infection.
\end{defi}

\begin{theorem}\label{forwardTracingRecursTheo}
We have $\kappa_{r,0}^+(a)=\widehat\kappa(a)$, and $\kappa_{r,i}^+(a)$ follows from the recursive equation
\begin{eqnarray}
\kappa _{r,i}^ + (a|b)\kappa _{r,i - 1}^ + (b) &=& \hat \kappa (a)\left\{ {\kappa _{r,i - 1}^ + (b) - p\int\limits_0^a {\left( { - \kappa {{_{r,i - 1}^ + }^\prime }(b + c) - \alpha \kappa _{r,i - 1}^ + (b + c)} \right)dc} } \right\},\label{forwardRekEqna}\\
\kappa_{r,i}^+ &=& \frac{\int_0^\infty\,\kappa_{r,i}^+(a|b)\, \kappa_{r,i-1}^+(b)\, db}{\int_0^\infty \kappa_{r,i-1}^+(\tau)\, d\tau}.
\label{forwardRekEqnb}
\end{eqnarray}
\end{theorem}
The proof of this theorem 2.8 is in wide parts identical with the analogue proof in \citep{muller2000contact,muller2016effect}; however, in order to keep the present paper self-contained, we briefly sketch the argument.
\begin{proof}
We first elaborate $\kappa_i^+(a|b)$, that is, the probability of a target individual in generation $i$ to be still infectious at age of infection $a$ if the infector did have age of infection $b$ at the time when the infection event did take place. The probability to be infectious at age of infection $a$, $\widehat\kappa(a)$, is decreased by contact tracing. We obtain the probability that no tracing event did take place. If the target individual has age $a$, the infector has age $a+b$. The probability that the infector is still infectious reads
$$
\frac{{\kappa _{r,i - 1}^ + (a + b)}}{{\kappa _{r,i - 1}^ + (a)}}
$$
The rate at which the infector is observed at age $b+c$ with $c \in [0,a)$ is given by
$$
\frac{\kappa_{r,i-1}^+(b + c)}
     {\kappa_{r,i-1}^+ (b)}\,\,
    \left( \frac{-{\kappa_{r,i-1}^+}'(b + c)}
    	        {\kappa_{r,i-1}^+(b + c)} - \alpha 
     \right) 
=
\left( 
    \frac{-{\kappa_{r,i-1}^+}'(b + c)}
         {\kappa_{r,i-1}^+(b)} 
  - \frac{\alpha \kappa_{r,i-1}^+(b + c)}
         {\kappa_{r,i-1}^+(b)} 
\right).
$$
Therefore, the probability that the infector induced a tracing event in $[0,a)$ sums up to
$$
\int\limits_0^a {\left( {\frac{{ - \kappa {{_{r,i - 1}^ + }^\prime }(b + c)}}{{\kappa _{r,i - 1}^ + (b)}} - \frac{{\alpha \kappa _{r,i - 1}^ + (b + c)}}{{\kappa _{r,i - 1}^ + (b)}}} \right)dc}.
$$
Therefore,
$$
\kappa _{r,i}^ + (a|b) = \hat \kappa (a)\left\{ {1 - p\int\limits_0^a {\left( {\frac{{ - \kappa {{_{r,i - 1}^ + }^\prime }(b + c)}}{{\kappa _{r,i - 1}^ + (b)}} - \frac{{\alpha \kappa _{r,i - 1}^ + (b + c)}}{{\kappa _{r,i - 1}^ + (b)}}} \right)dc} } \right\},
$$
or, equivalently, 
$$
\kappa _{r,i}^ + (a|b)\kappa _{r,i - 1}^ + (b) = \hat \kappa (a)\left\{ {\kappa _{r,i - 1}^ + (b) - p\int\limits_0^a {\left( { - \kappa {{_{r,i - 1}^ + }^\prime }(b + c) - \alpha \kappa _{r,i - 1}^ + (b + c)} \right)dc} } \right\}.
$$
We find $\kappa_{r,i}^+(a)$ in integrating $\kappa_{r,i}^+(a|b)$ by the probability density of a generation $i$-infected to produce the secondary case at age of infection $b$. Thereto, we focus on the edge connecting the focal individual (generation $i$) with his/her upstream node (generation $i-1$). If the upstream node becomes infected, the focal node is uninfected. The infection rate for the focal node (note: we do not condition on the fact that the upstream node is a certain time interval infected, as we did in Theorem~\ref{backRecurs}) is 
$ \beta\,\kappa_{r,i-1}^+(b)$. 
Hence, the age of infection at the time of the infection event is given by the probability distribution
$\beta\kappa_{r,i-1}^+(b)/\int_0^\infty\beta\kappa_{r,i-1}^+(\tau)\, d\tau$.
This consideration yields eqn.~(\ref{forwardRekEqnb}), and completes the proof.
\end{proof}

\begin{rem}
We may rewrite the integral for $\kappa_{r,i}^+(a)$ as follows:
\begin{eqnarray*}
\kappa _{r,i}^ + (a) &=& \frac{{\int\limits_0^\infty  {\kappa _{r,i}^ + (a|b)\kappa _{r,i - 1}^ + (b)db} }}{{\int\limits_0^\infty  {\kappa _{r,i - 1}^ + (\tau)d\tau} }}
= \frac{{\int\limits_0^\infty  {\hat \kappa (a)\left( {\kappa _{r,i - 1}^ + (b) - p\int\limits_0^a {\left( { - \kappa {{_{r,i - 1}^ + }^\prime }(b + c) - \alpha \kappa _{r,i - 1}^ + (b + c)} \right)dc} } \right)} db}}{{\int\limits_0^\infty  {\kappa _{r,i - 1}^ + (\tau)d\tau} }}\\
&=&\widehat \kappa (a)\left( {1 - p\frac{{\int\limits_0^\infty  {\int\limits_0^a {\left( { - \kappa {{_{r,i - 1}^ + }^\prime }(b + c) - \alpha \kappa _{r,i - 1}^ + (b + c)} \right)dcdb} } }}{{\int\limits_0^\infty  {\kappa _{r,i - 1}^ + (\tau)d\tau} }}} \right).
\end{eqnarray*}
Since
$$
\int_0^\infty \int_0^a f(b+c)\,dcdb 
=
\int_0^\infty \int_b^{a+b} f(c)\,dcdb 
= 
\int_0^a \int_0^c f(c)\,dbdc 
+
\int_a^\infty \int_{c-a}^{c} f(c)\,dbdc 
= \int_0^\infty \min\{a,c\}\, f(c)\, dc
$$
we find 
\begin{equation} \label{eq:37}
\kappa _{r,i}^ + (a) = \hat \kappa (a)\left( {1 - p\frac{{\int\limits_0^\infty  \min\{a,b\}\,\bigg(-\kappa {_{r,i - 1}^{+ \prime} }(b) - \alpha \kappa _{r,i - 1}^+ (b)\bigg)\, db}}{{\int\limits_0^\infty  {\kappa _{r,i - 1}^ + (\tau)d\tau} }}} \right).
\end{equation}
\end{rem}

\begin{figure}[h!]
	\centering
	\includegraphics[width=\textwidth]{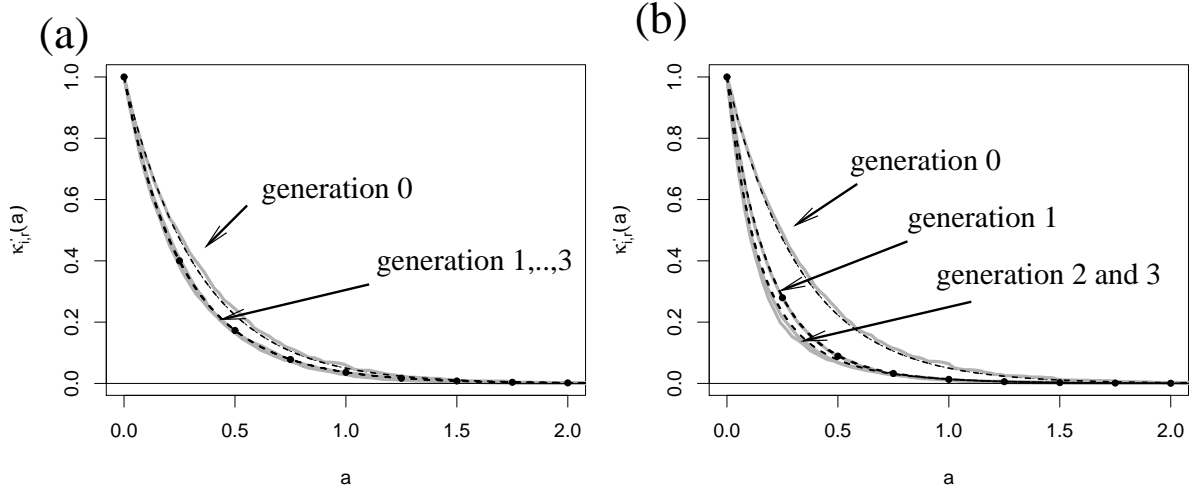}
	\caption{ Recursive forward tracing. The probability to be infectious at age $a$ after infection ($\kappa_{r,i}^+(a)$) for $p=0.3$ (panel a), and $p=0.8$ (panel b). Dashed lines: numerical solution of $\kappa_{r,i}^+(a)$; dashed-dotted  line: without tracing
	($\hat \kappa (a)=e^{-(\sigma+\alpha)a}$) coinciding with $\kappa_{r,0}^+(a)$;  dashed-lines with bullets: first order approximation of $\kappa_{r,i}^+(a)$; grey line: simulation. Parameters: $\beta=1.5$, $\alpha = 0.1$, $\sigma=2.9$, $E(K)=4$, $K$ is constant.}
	\label{fig:ForwardRec}
\end{figure}

\begin{prop}\label{ForwardRECURSfIRSTOrder}
Let $p_{obs}=\sigma/(\sigma+\alpha)$. The first order approximation of ${\kappa _{i}^ + (a)}$ reads
\begin{equation} \label{eq:50}
\kappa _{r,i}^ + (a) = \hat \kappa (a)\left( {1 - p{p_{obs}}\left( {1 - \hat \kappa (a)} \right)} \right) + \mathcal{O} ({p^2}).
\end{equation}
\end{prop}
\begin{proof} As above, in the first order approximation for backward tracing, we note that we only need the zero order approximation of $\kappa_{r,i-1}$ if that probability is multiplied by $p$. Therefore, 
\begin{align*}
\kappa _{r,i}^ + (a) ={}& \hat \kappa (a)\left( {1 - p\frac{{\int\limits_0^\infty  {\int\limits_0^a {\sigma\,e^{-(\alpha+\sigma)(b + c)}dcdb} } }}{{\int\limits_0^\infty  {e^{-(\alpha+\sigma)\,b}\,db} }}} \right) + \mathcal{O} ({p^2})
=  \hat \kappa (a)\left( {1 - p{p_{obs}}\left( {1 - \hat \kappa (a)} \right)} \right) + \mathcal{O} ({p^2}).
\end{align*}
\end{proof}
As before, the theory and simulations agree nicely as shown in Fig.~\ref{fig:ForwardRec}.

\subsubsection{One-step method}

Similar to backward tracing, the difference between the recursive method and the one-step method is that in the one-step method, an infected individual discovered via tracing cannot trigger another tracing event.

\begin{theorem}
 For the probability that an individual reaches age $a$ of infection in forward tracing with the one-step method, $\kappa _{o,i}^+(a)$, the following applies:
\begin{equation} \label{eq:57}
\kappa _{o,i}^+(a)  = \hat \kappa (a)\left( {1 - p\frac{{\int\limits_0^\infty  {\int\limits_0^a {\sigma \kappa _{o,i - 1}^{+} (b + c)dcdb} } }}{{\int\limits_0^\infty  {\kappa _{o,i - 1}^{+} (b)db} }}} \right).
\end{equation}
\end{theorem}
\begin{proof}
The proof follows theorem 2.8 by replacing $- \kappa {_{r,i - 1}^{+ \prime} }(a) - \alpha \kappa _{r,i - 1}^+ (a)$ with $- \sigma \kappa _{o,i - 1}^{+} (a)$ for the reason stated above.
\end{proof}

Recall that $p_{obs}=\sigma/(\sigma+\alpha)$ is the probability to be diagnosed (if no contact tracing takes place). With this definition, we find the following result. 
\begin{prop} The first order approximation for $\kappa {_{o,i}^{+ \prime }}(a)$ reads

\begin{equation} \label{eq:59}
\kappa _{o,i}^+ (a) = \hat \kappa (a)\left( {1 - p{p_{obs}}\left( {1 - \hat \kappa (a)} \right)} \right) + \mathcal{O} ({p^2}).
\end{equation}
\end{prop}
\begin{proof}
See proof of Proposition~\ref{ForwardRECURSfIRSTOrder}.
\end{proof}

As expected, the first order approximations of $\kappa _{o,i}^ + $ and $\kappa _{r,i}^ +$ coincide because the difference between the recursive and the one-step method results only from the second tracing step. This comes with a probability of $\mathcal{O} ({p^2})$, thus the difference between the two methods in the approximation is expressed only in the remainder of the term. The first order approximations of $\kappa _{r,i}^ + $ and $\kappa _{o,i}^ +$ are independent of the generation of the individual in consideration and also agree with the given first order approximation of \citep{muller2016effect}. Again, this is not surprising given that the formulas for $\kappa _{r,i}^ + $ and $\kappa _{o,i}^+$ have been developed using the same approach. A comparison of theoretical results and simulations is presented in Fig.~\ref{fig:ForwardOneStep}.

\begin{figure}[h!]
	\centering
	\includegraphics[width=\textwidth]{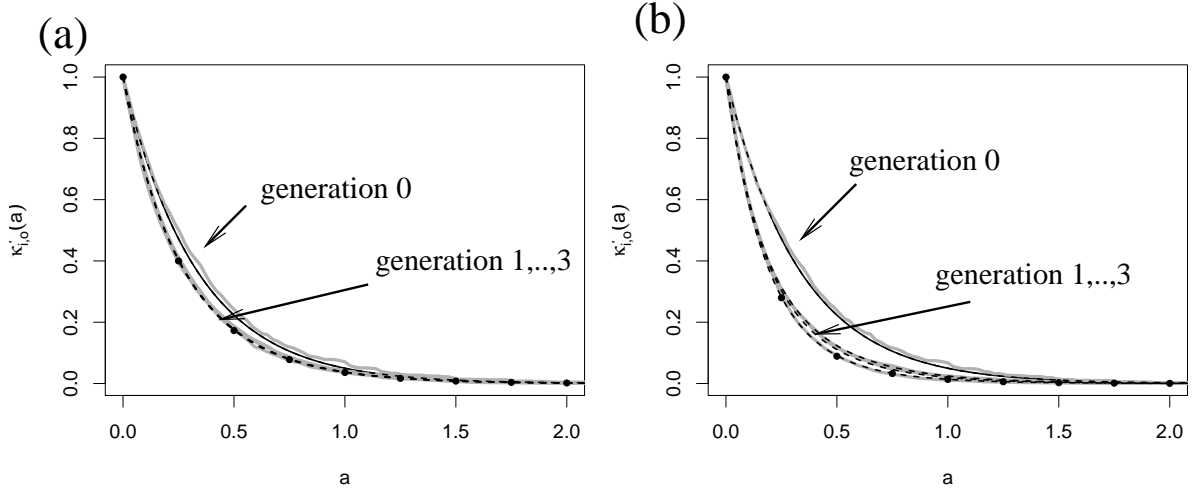}
	\caption{ One-step forward tracing. The probability to be infectious at age $a$ after infection ($\kappa_{o,i}^+(a)$) for $p=0.3$ (panel a), and $p=0.8$ (panel b). 
	Dashed lines: numerical solution of $\kappa_{o,i}^+(a)$; solid line: without tracing ($\hat \kappa (a)=e^{-(\sigma+\alpha)a}$) coinciding with $\kappa_{o,0}^+(a)$; dashed-lines with bullets: first order approximation of $\kappa_{o,i}^+(a)$; grey line: simulation. Parameters: $\beta=1.5$, $\alpha = 0.1$, $\sigma=2.9$, $E(K)=4$, $K$ is constant.}
	\label{fig:ForwardOneStep}
\end{figure}

\subsection{Full tracing}
In the following, a formula for full tracing, ${\kappa _{r,i}(a)}$, will be developed. This is precisely the probability that an infected person reaches the age of infection a when full tracing is applied. For this, the formula for forward and backward tracing set up in the previous sections must be suitably combined.
 
\begin{theorem}
We have $\kappa_{r,0}(a)=\kappa_r^-(a)$, and 
for $i>0$ 
\begin{equation}
{\kappa _{{r,i}}}(a) = \frac{{\int\limits_0^\infty  {{\kappa _r^-}(a)\left( {{\kappa _{r,i - {1}}}(b) - p\int\limits_0^a { - {{\kappa '}_{r,i - {1}}}(b + c) - \alpha } {\kappa _{r,i - 1}}(b + c)dc} \right)} db}}{{\int\limits_0^\infty  {{\kappa _{r,i - {1}}}(b)db} }},
\end{equation}
where  $\kappa _r^-(a)$ is given by eqn.~(\ref{eq:2}).
\end{theorem}
\begin{proof} The proof parallels the proof of theorem~\ref{forwardTracingRecursTheo} (forward tracing), where we use that $\kappa_{r,i}(a)$ equals the probability to be still infectious at age $a$ under the condition that  the individual was not hit by a forward tracing event, times the probability that indeed the individual was not removed by backward tracing. As we now allow for backward tracing the probability to be infectious at age conditioned on no forward tracing is given by $\kappa_r^-(a)$. Otherwise, the argument remains the same. 
\end{proof}

\begin{figure}[h!]
	\centering
	\includegraphics[width=\textwidth]{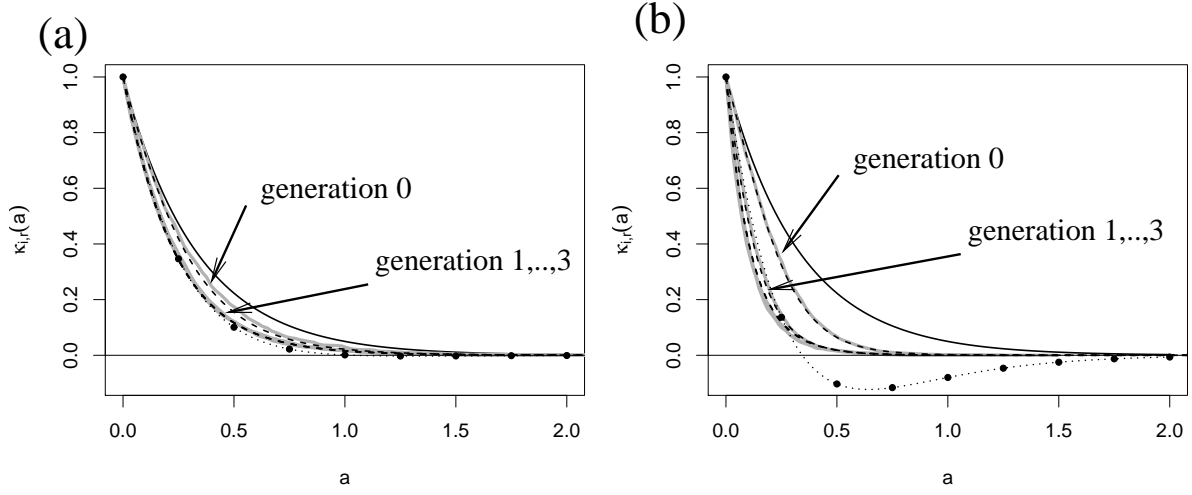}
	
	\caption{Recursive full tracing. The probability to be infectious at age $a$ after infection ($\kappa_{r,i}(a)$) for $p=0.3$ (panel a), and $p=0.8$ (panel b).  Solid line: without tracing
	($\hat \kappa (a)=e^{-(\sigma+\alpha)a}$); dashed lines: numerical solution of $\kappa_{r,i}(a)$; dashed-dotted lines with bullets: first order approximation of $\kappa_{r,i}(a)$; grey line: simulation. 
	Parameters: $\beta=1.5$, $\alpha = 0.1$, $\sigma=2.9$, $E(K)=4$, $K$ is constant.}
	
	\label{fig:FullRec}
\end{figure}

The following proposition can be shown along the lines for the corresponding approximations in forward- and backward tracing; the simulations confirm also in that case that we catch the heart of the process (Fig.~\ref{fig:FullRec}).

\begin{prop}
The first order approximation of ${\kappa _{{i}}(a)}$ reads
\begin{equation}\label{fullRecApprox}
{\kappa _{r,i}}(a) = {\hat \kappa (a)}\left( {1 + \frac{{p\sigma E(K)}}{{\alpha  + \sigma  - \beta }}\left( {{e^{ - \beta a}} - {\hat \kappa (a)}} \right) - \frac{{p\sigma E(K)}}{{\alpha  + \sigma }}\left( {1 - {\hat \kappa (a)}} \right) - \frac{{p\sigma }}{{\alpha  + \sigma }}\left( {1 - {\hat \kappa (a)}} \right)} \right) + {\cal O}({p^2}).
\end{equation}
\end{prop}

We find the parallel results for one-step tracing. Simulations are presented in Fig.~\ref{fig:FullOneStep}.
\begin{theorem}
We have $\kappa_{o,0}(a)=\kappa_o^-(a)$, and 
for $i>0$ 
\begin{equation}
\kappa _{o,i}(a) 
= 
\frac{\int\limits_0^\infty \,\kappa _o^-(a)\left( \kappa _{o,i - 1}(b) - p\,\sigma\,\int\limits_0^a    \kappa _{o,i - 1}(b + c)\,dc \right) db}
{\int\limits_0^\infty  \kappa _{o,i - 1}(b)db },
\end{equation}
where  $\kappa _o^-(a)$ is given by eqn.~(\ref{eq:28}). The first order approximation is identical with that for $\kappa_{r,i}^-(a)$ in eqn.~(\ref{fullRecApprox}).
\end{theorem}

\begin{figure}[h!]
	\centering
	\includegraphics[width=\textwidth]{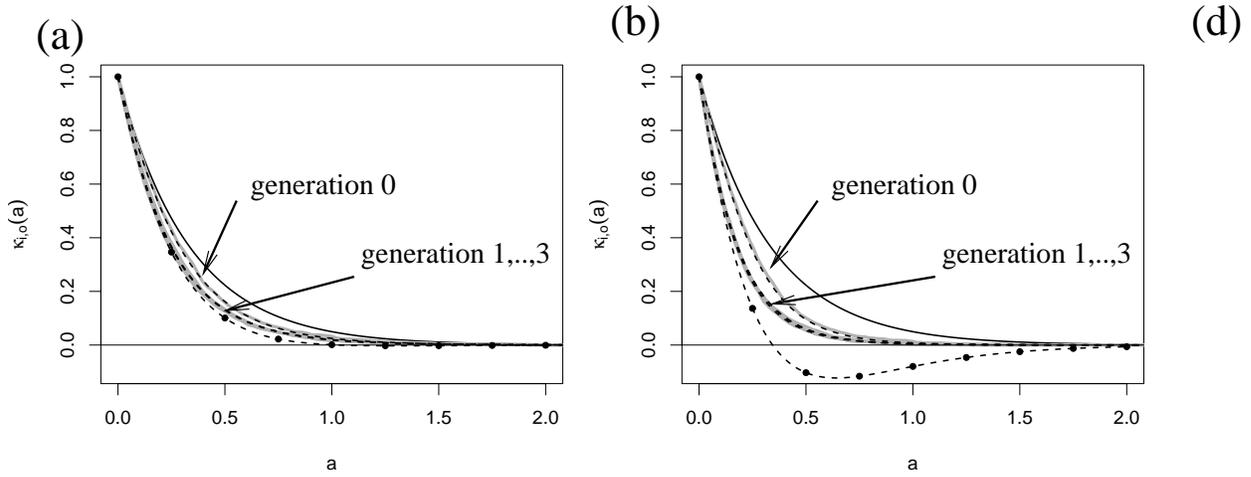}
	\caption{One-step full tracing. The probability to be infectious at age $a$ after infection ($\kappa_{o,i}(a)$) for $p=0.3$ (panel a), and $p=0.8$ (panel b).  Solid line: without tracing
	($\hat \kappa (a)=e^{-(\sigma+\alpha)a}$); dashed lines: numerical solution of $\kappa_{o,i}(a)$; dashed-lines with bullets: first order approximation of $\kappa_{o,i}(a)$; grey line: simulation. 
	Parameters: $\beta=1.5$, $\alpha = 0.1$, $\sigma=2.9$, $E(K)=4$, $K$ is constant.}
	\label{fig:FullOneStep}
\end{figure}

\begin{rem}
We understand that the structure of forward tracing does not differ from the present case and a homogeneously mixing population. For backward tracing, we have the approximation results (Proposition~\ref{approxBackwardHomo1} and \ref{approxBackwardHomo2}). Therefore, those approximation results directly carry over to full tracing. Particularly, if we consider trees where $K$ is distributed according to a poisson distribution, and $\beta$ is small, s.t.\ $E(K)\rightarrow \infty$, $\beta\rightarrow 0$, while $\beta\,E(K)\rightarrow \beta_{ind}$, then $\kappa_{r,i}(a)$ tends to the corresponding probability for full tracing in a homogeneous population (if we consider the onset of the epidemic). The parallel result holds true for one-step tracing.
\end{rem}

\subsection{Reproduction number}
The reproduction number is defined as the average number of secondary infections generated by one infectious individual in a population of completely susceptible individuals \citep{anderson1992infectious, diekmann1995legacy}.
Without contact tracing ($p=0$), the basic reproduction number is given by 
\begin{eqnarray}
R_0&=& \int\limits_0^\infty  {\beta \left( {\sum\limits_{k = 0}^\infty  {P(K = k)} }\hat \kappa (a) \right)da} 
= \int\limits_0^\infty  {\beta E(K)\hat \kappa (a)da}
= \frac{{\beta E(K)}}{{\alpha  + \sigma }}.
\end{eqnarray}
This reproduction number is decreased by contact tracing; we denote the reproduction number with contact tracing by $R_{ct}$. 

\begin{theorem}
Under full tracing (either recursive or one-step), the reproduction number of the ith generation reads 

\begin{equation} \label{eq:60}
R_{ct} = R_0\Bigg( 1 - \frac{{p\sigma \beta E{{(K)}}}}{{{{2(\alpha+\sigma)(\alpha  + \sigma + \beta)}}}} - \frac{{p\sigma }}{{2(\alpha  + \sigma )}} \Bigg) + {\cal O}(p^2).
\end{equation}

\end{theorem}

\begin{proof}
For $p>0$, we use the first order approximation of $\kappa_{\ast,i}(a)$ and find
\begin{align*}
R_{ct} ={}& \int\limits_0^\infty  {\beta \left( {\sum\limits_{k = 0}^\infty  {P(K = k)} } {\kappa _{\ast,i}}(a)\right)da},\\
 ={}& \int\limits_0^\infty  {\beta E(K)\hat \kappa (a)\left[ {1 + \frac{{p\sigma E(K)}}{{\alpha  + \sigma  - \beta }}\left( {{e^{ - \beta a}} - \hat \kappa (a)} \right) - \frac{{p\sigma E(K)}}{{\alpha  + \sigma }}\left( {1 - \hat \kappa (a)} \right) - \frac{{p\sigma }}{{\alpha  + \sigma }}\left( {1 - \hat \kappa (a)} \right)} \right]da + {\cal O}({p^2})}
\end{align*}
The result follows with
$
\int\limits_0^\infty  {\hat \kappa (a)\left( {{e^{ - \beta a}} - \hat \kappa (a)} \right)da}={} \frac{\alpha-\sigma-\beta}{2(\alpha  + \sigma ){(\alpha  + \sigma  + \beta)}}$
and
$\int\limits_0^\infty  {\hat \kappa (a)\left( {1 - \hat \kappa (a)} \right)da} ={} \frac{1}{{2(\alpha  + \sigma )}}.
$
\end{proof}

\begin{rem} We again consider the limit $\beta\rightarrow 0$, $E(K)\rightarrow\infty$, while $\beta(E(K)+1)\rightarrow\beta_{ind}$. 
With $R_0=\beta_{ind}/(\sigma+\alpha)$ and $p_{obs}=\sigma/(\alpha+\sigma)$, we obtain the limit
\begin{eqnarray}
R_{ct} \rightarrow R_0\bigg(1-\,\frac1 2 \,\,p\,p_{obs}\, (R_0 + 1)\bigg) + \mathcal{O}(p^2),
\end{eqnarray}
which is identical with the first order approximation of the reproduciton number obtained in a randomly mixing model~\cite{muller2000contact,muller2016effect}.
\end{rem}

\subsection{Influence of the Graph structure}

The structure of the random tree, given by the generating function of the downstream degree distribution $G(s)$, only appeared in the consideration of backward tracing (see Theorem~\ref{backRecurs} and Theorem~\ref{backOneStep}). As full tracing is based on forward- and on backward tracing, also there, the graph structure comes in. 
However, all first order approximations did only depend on $E(K)$. If we expand $\kappa^-_*(a)$ w.r.t.\ $p$, a closer look shows that the $i$'th order correction term (proportional to $p^i$) depends on the first $i$ moments of 
$K$. As in applications, $p$ is rather small, we expect indeed the special choice of the random tree only has a minor influence, if all moments of $K$ are finite. 

\begin{figure}[h!]
	\begin{center}
		\includegraphics[width=7cm]{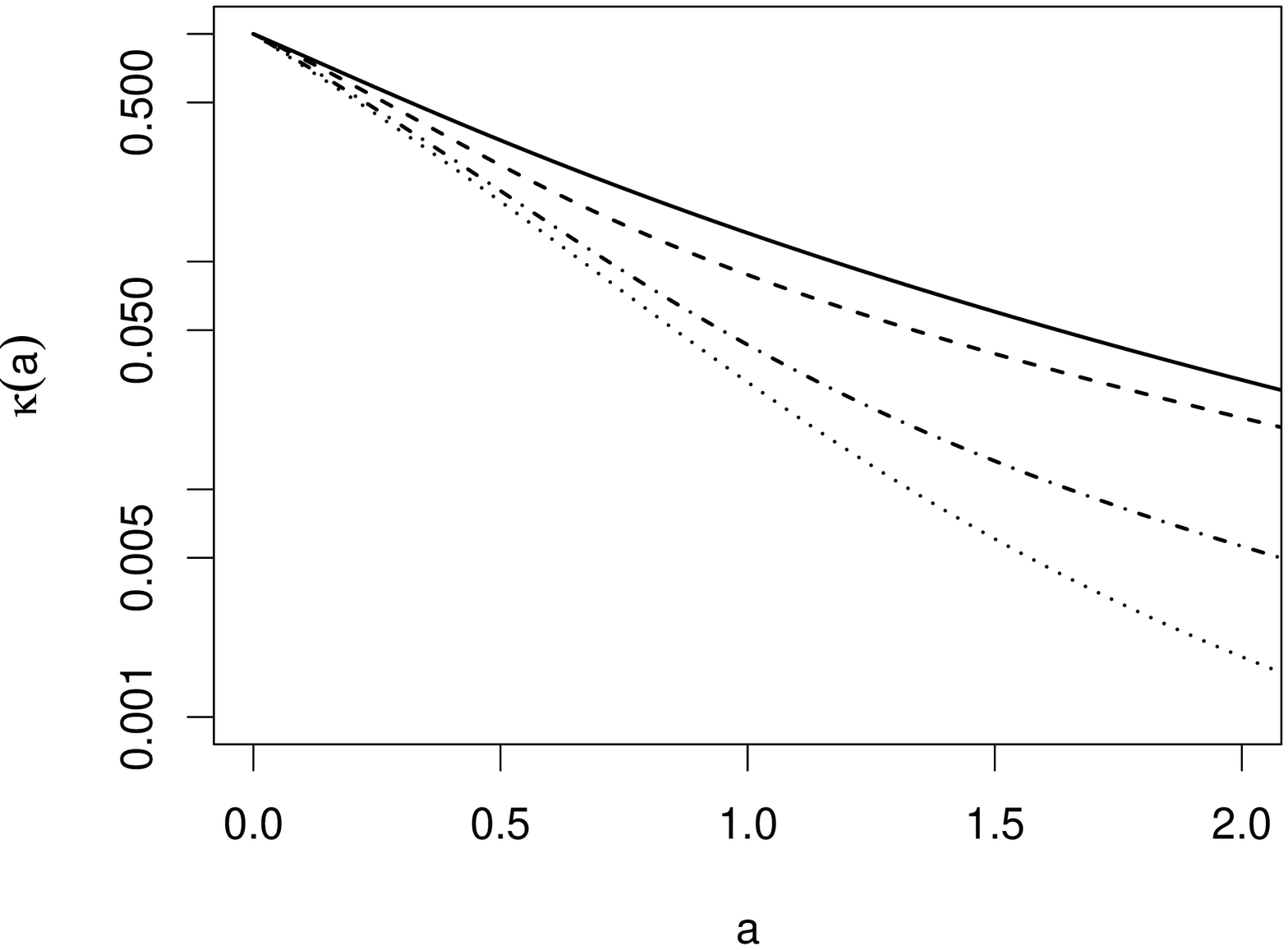}
		\includegraphics[width=7cm]{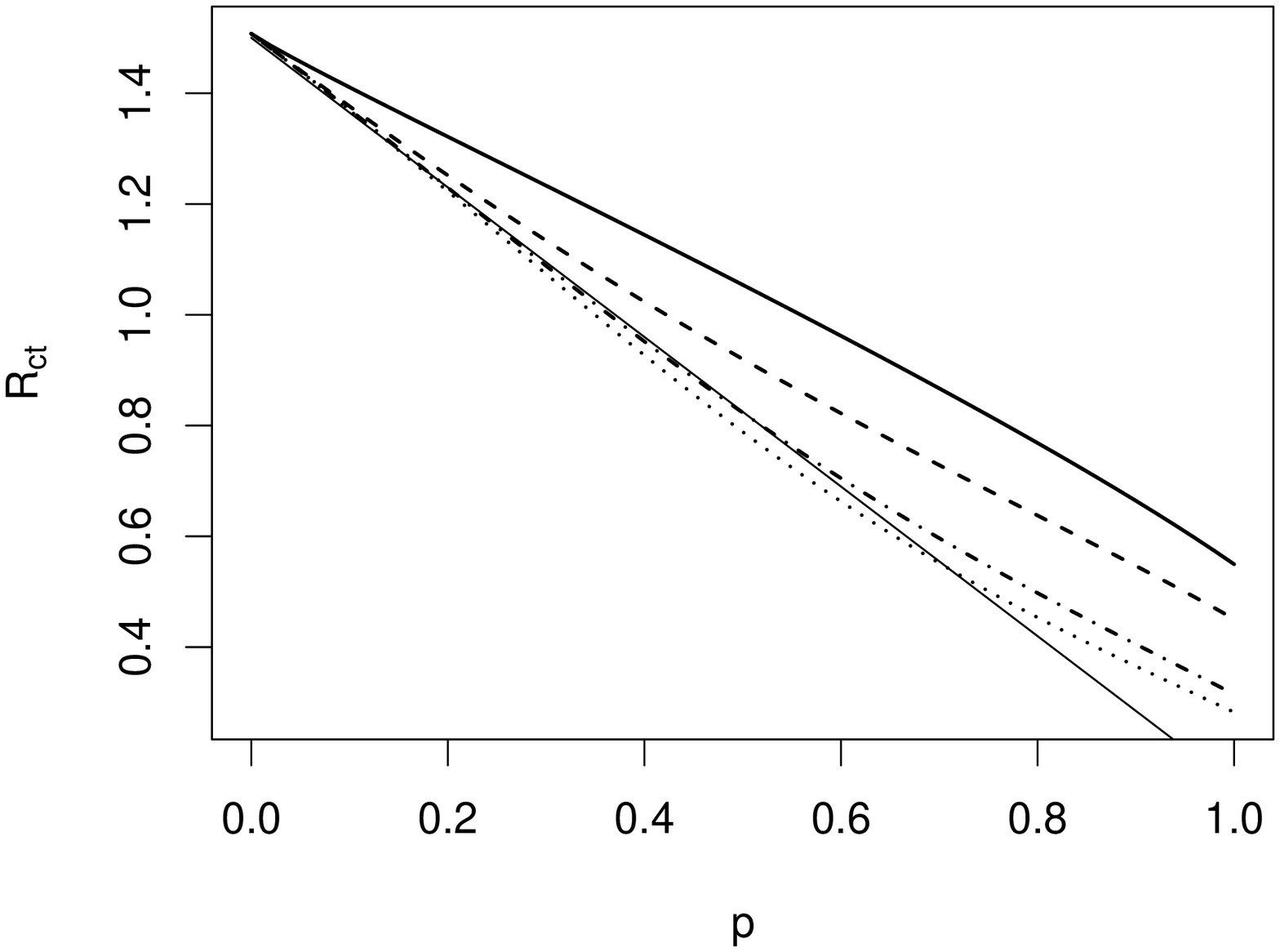}
	\end{center}
	\caption{Influence of the distribution of $K$. Left panel: $\kappa_{r,5}(a)$ (in semi-logarithmic representation), right panel: $R_{ct}$ over p. 
		Curves, from the top down: solid, $K$ with large second moment (see text, variance  $4064$); dashes, geometric distribution (variance 23.11); dashed-dotted,  Poisson distribution (variance 4.3); dotted, $K=$ constant (variance 0). Additional (thin, solid) line in the right panel: First order approximation of $R_0$. 
		Parameters:	$E(K)=4.3$, $\beta=1.5$, $\alpha=\sigma=0.5$. }
	\label{treeVar}
\end{figure}

A large or even infinite second moment, however, will lead to different results. To exemplify this conjecture, we consider four different random variables $K$: (a) $K$ is constant, (b) $K$ is Poisson distributed, (c) $K$ geometric distributed, and (d) $K$ follows a truncated power law, 
$P(K=i)=c\, i^{-\gamma}$ for $i\in\{1,\ldots, 15000\}$, and $P(K=i)=0$ else, where we choose $\gamma=2.1$. Clearly if we do not restrict ourselves to a $K \leq 15000$, but extend the distribution for all positive integers, the expectation is finite, while the variance is not. For numerical reasons, we exclude in this last case the large numbers of $K$ and truncate the power law.\\ 
In the simulations, we adapt the parameters such that the expectations of the four distributions are identical, while the variance is increasing from (a) to (d) -- see caption of Fig.~\ref{treeVar} for the exact numbers. 
We find indeed that the effect of contact tracing decreases with the variance of $K$ (Fig.~\ref{treeVar}). Moreover, the first order 
approximation for $R_0$ is only appropriate if the variance is not too large; even for the geometric distribution, there is some reasonable deviation for $p>0.3$, while for the distribution (d), even small values of $p$ show a different behaviour. Nevertheless, at $p=0$, approximation and exact result are tangential to each other, as predicted by the theory. 
Our finding bears some similarity with the observation that the reproduction number monotonously depends on the variance of the inter-generation time~\citep{Wallinga2006}. In that case, the larger the variance the smaller the reproduction number.

\section{Mean field}
So far, we did mainly focus on single individuals. The fate of an individual is well characterized by the infectious interval, described by $\kappa(a)$. Now we turn to the population. 
Preliminary studies (\cite{muller2000contact,muller2016effect}) did show that in homogeneous populations, the local correlations caused by contact tracing (removal of several neighbouring infected individuals at the same time) are less important than the reduction of the infectious period. We also test this observation in the present setting.\\
The situation we did analyse is a tree, where initially the root is infected. 
We develop an idea about the dynamics of the prevalence in the next section. In the section thereafter, we extend our ideas by means of heuristic arguments to the configuration model

\subsection{Mean field approximation on a tree}

The dynamics of an SIR model without contact tracing on a tree is well known \citep{Diekmann1998,keeling2005networks}. We briefly sketch the idea of the analysis. 
Let $I_t$, $R_t$ the number of infected/recovered nodes at time $t$. Clearly,
$$\frac d {dt} E(R_t) = \gamma E(I_t).$$
It remains to focus on $I$. Consider an infected individual at age of infection $a$ and $k$ downstream nodes. Furthermore, let $z_k(a)$ denote the expected number of secondary infecteds. Clearly, $z_k(0)=0$, $z_k(a)\leq k$. We have 
$$\frac d {da}z_k(a) = \, \beta\, (k-z_k(a)),\qquad z_k(0)=0.$$
Hence, $z_k(a) = k\left(1-e^{-\beta a}\right)$.
If we want to know the number of children of a randomly selected individuals that became infected $a$ time units ago, 
we need to remove the information about $K$, 
$$z(a) = \sum_{k=0}^\infty k\left(1-e^{-\beta a}\right)\, P(K=k) = E(K)\, \left(1-e^{-\beta a }\right).
$$
Hence the rate $\theta(a)$ at which an infected individual produces downstream infecteds is given by 
$$\theta(a) = \frac d {da}z(a) =E(K)\, \beta\,e^{-\beta a}.
$$
The last term, $\beta e^{-\beta a}$ 
refers to the time distribution of the first infectious event, while $E(K)$ is the average number of nodes that can be infected. Using these considerations, we are able to set up an age-structured model for the density of infected individuals, 
\begin{eqnarray*}
(\partial_t+\partial_a) i(t,a) &=& -\gamma i(t,a)\\
i(t,0) &=& \int_0^\infty \theta(a)\, i(t,a)\, da.
\end{eqnarray*}
A further, interesting subpopulation is the number of missed downstream neighbours. If the upstream neighbour of a susceptible is recovered, that individual escaped the infection. 
Let $\Sigma(t)$ be the expected number of these individual. 
We find 
\begin{eqnarray}
(\partial_t+\partial_a) i(t,a) &=& -\gamma i(t,a)\label{treeMeanMod1}\\
i(t,0) &=& \int_0^\infty \theta(a)\, i(t,a)\, da \label{treeMeanMod2}\\
\frac d {dt} R &=& \gamma \int_0^\infty i(t,a)\, da\label{treeMeanMod3}\\
\frac d {dt} \Sigma &=& \gamma \int_0^\infty i(t,a)\,(E(K)-z(a))\, da
=  E(K)\,
\gamma \int_0^\infty i(t,a)\,e^{-\beta a}\, da\label{treeMeanMod4}
\end{eqnarray}
As usual, this age structured model will tend to an exponential growing solution with a stable age structure,
$$i(t,a) = I_0\,e^{\lambda t}\, i(a).$$ 
Therewith,
$$ i'(a) = -(\gamma+\lambda) i(a)
\quad\Rightarrow\quad
i(a) = e^{-(\gamma+\lambda)\, a}.
$$
The boundary condition for $i(t,0)$ yields 
$$ 1 
= \int_0^\infty \theta(a)\, e^{-(\lambda+\gamma)a} \, da
=  E(K)\,\beta\,\, \int_0^\infty \,\,e^{-(\lambda+\gamma+\beta)a} \, da
=
\frac{\beta E(K)}{\lambda+\gamma+\beta}
.
$$
Solving this expression for $\lambda$, we have 
$$\lambda = \beta (E(K)-1)-\gamma.$$
In the long run, $i(t,a) = I(t)\,\lambda\, e^{-(\lambda+\gamma) \, a}$ and hence 
\begin{eqnarray*}
	\frac d {dt} \Sigma &=&   E(K)\,
	\gamma I\,\int_0^\infty\,(\lambda+\gamma)\,e^{(-\lambda-\gamma-\beta) a}\, da\nonumber\\
	&=& E(K)\,\frac{\gamma+\lambda}{\lambda+\gamma +\beta}\, \gamma\,I
	= E(K)\, \frac{ \beta (E(K)-1)}{\beta E(K)}\, \gamma\,I
	= (E(K)-1)\, \gamma\,I
\end{eqnarray*}

\begin{cor} In the long run, we find 
	\begin{eqnarray}
	\frac d {dt} I &=& \beta (E(K)-1)\, I-\gamma I\\
	\frac d {dt} R &=& \gamma I\\
	\frac d {dt} \Sigma
	&=&  (E(K)-1)\, \gamma\,I
	\end{eqnarray}
\end{cor}
A short computation shows that, in the long run, the expected number of susceptible downstream nodes of a randomly selected infected node is $E(K)-1$ ~\citep{Diekmann1998,keeling2005networks}. Hence, the number of infected grow exponentially fast, where the incidence is $\beta(E(K)-1)$, and the recovery rate $\gamma$, such that the exponent reads $\lambda = \beta (E(K)-1)-\gamma$. 
\par\medskip 
We compute the asymptotic ratios
$\lim_{t\rightarrow\infty} \frac{\Sigma(t)}{I(t)}$, 
$\lim_{t\rightarrow\infty} \frac{R(t)}{I(t)}$. 
With $\zeta(t) = \Sigma(t)/I(t)$ and $\eta(t)=R(t)/I(t)$, 
we find
\begin{eqnarray*}
	\zeta'(t) &=& 
	\frac{\Sigma'(t)}{I(t)}- \frac{\Sigma(t)}{I(t)}\, \frac{I'(t)}{I(t)}
	=
	(E(K)-1)\gamma-  \lambda\,\zeta(t)\\
	\eta'(t) &=& 
	\frac{R'(t)}{I(t)}- \frac{R(t)}{I(t)}\, \frac{I'(t)}{I(t)}
	=
	\gamma - \lambda \,\eta(t)
\end{eqnarray*}
and hence
\begin{cor}
	$$\lim_{t\rightarrow\infty} \frac{\Sigma(t)}{I(t)}
	=  (E(K)-1)\frac{\gamma}{\lambda}
	=(E(K)-1)\frac{\gamma}{\beta(E(K)-1)-\gamma},
	\qquad
	\lim_{t\rightarrow\infty} \frac{R(t)}{I(t)}
	=\frac{\gamma}{\beta(E(K)-1)-\gamma}.
	$$
\end{cor}

In order to incorporate contact tracing, we replace in eqn.~(\ref{treeMeanMod1})-(\ref{treeMeanMod2}), the constant removal rate $\gamma$ by the age structured non-constant hazard rate $-\kappa'(a)/\kappa(a)$, where $\kappa(a)$ is the probability to be infectious at age of infection $a$ (in case of forward/full contact tracing, the asymptotic probability, if we take the number of generations to $\infty$). Using the first order approximation of $\kappa(a)$ in eqn. \ref{fullRecApprox}, that is $\kappa_{r,i}(a)$, the age-structured model now becomes
\begin{eqnarray*}
(\partial_t+\partial_a) i(t,a) &=& log(\kappa(a))' i(t,a)\\
i(t,0) &=& \int_0^\infty \theta(a)\, i(t,a)\, da.
\end{eqnarray*}
Asymptotically, $i(t,a) = I_{0}\,e^{\lambda t}\,i(a)$, where 
$$i'(a) = (log(\kappa(a))' - \lambda)\,i(a).$$
Thus $$i(a) = i(0)\,\kappa(a)\,e^{-\lambda a}$$
and the equation that determines $\lambda$ becomes
$$ 1 = \int_0^\infty  {\theta (a)\,\kappa (a)\,{e^{ - \lambda a}}\,da}.$$

A short computation of the asymptotic ratios which incorporates contact tracing yields

\begin{cor}
\[\mathop {\lim }\limits_{t \to \infty } \frac{{\Sigma (t)}}{{I(t)}} = \frac{{E(K)\left( {1 - (\lambda  + \beta )\int_0^\infty  {\kappa (a){e^{ - (\lambda  + \beta )a}}}\,da} \right)}}{{\lambda \int_0^\infty  {\kappa (a){e^{ - \lambda a}}}\,da}},
\qquad
\mathop {\lim }\limits_{t \to \infty } \frac{{R(t)}}{{I(t)}} = \frac{1}{{\lambda \int_0^\infty  {\kappa (a){e^{ - \lambda a}}}\,da}} - 1.\]
\end{cor}


For different tracing probabilities, the effect of contact tracing on the exponent and asymptotic ratios is established, and with the age structured non-constant hazard rate, we find a satisfying agreement with simulation results (Fig. \ref{fullOneStepExponentialGrowth}).  It is not surprising that for increased $p$, the prediction for the non-age structured removal rate ($\lambda = \beta (E(K) - 1) - 1/{\int_0^\infty  {\kappa (a)da} }$) breaks down, this is because we have used a first-order approximation of $\kappa(a)$ to compute $\lambda$.

\begin{figure}[!h]
\centering
  $\vcenter{\hbox{\includegraphics[height=5cm]{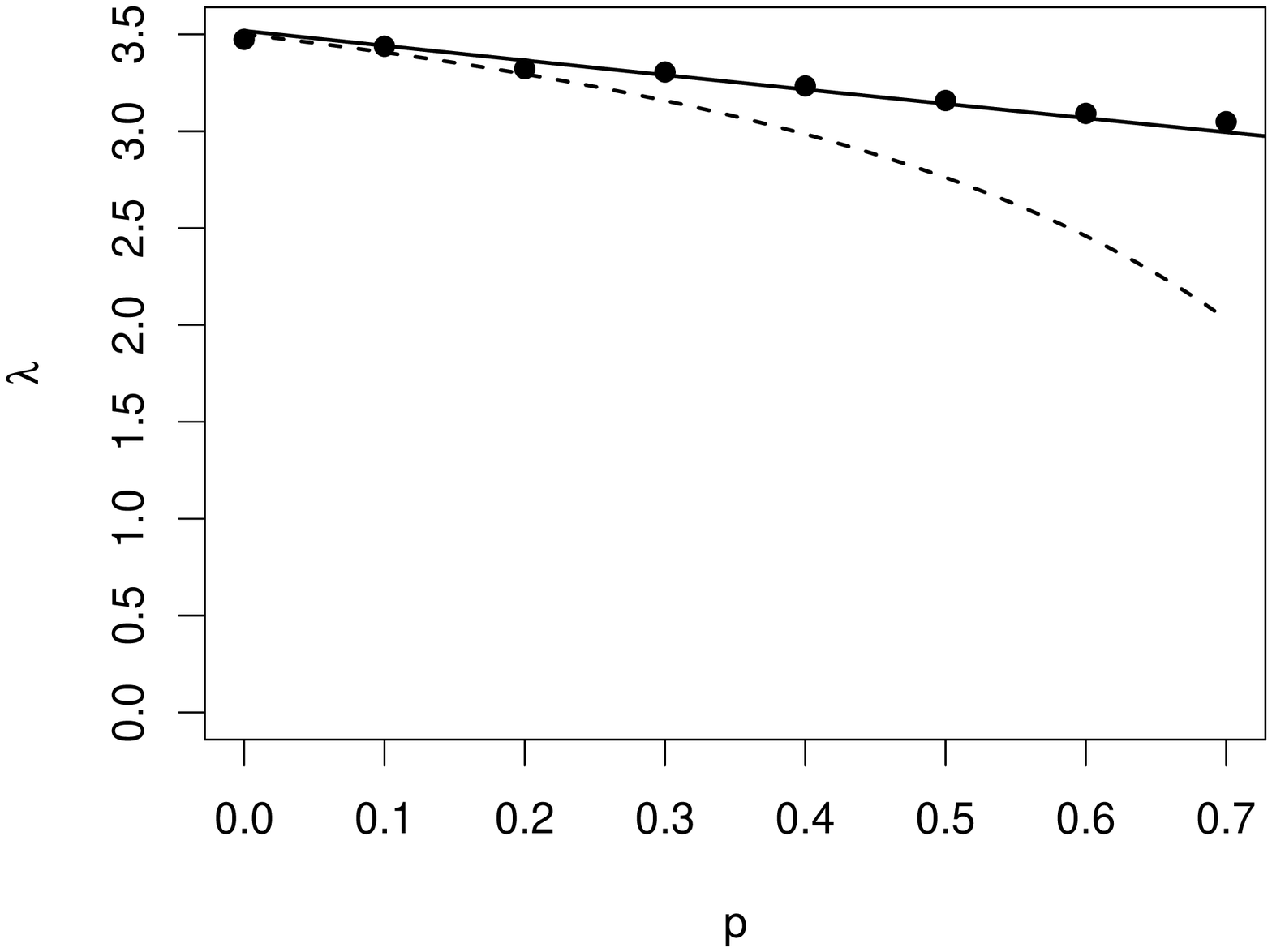}}}
  \vcenter{\hbox{\includegraphics[height=5cm]{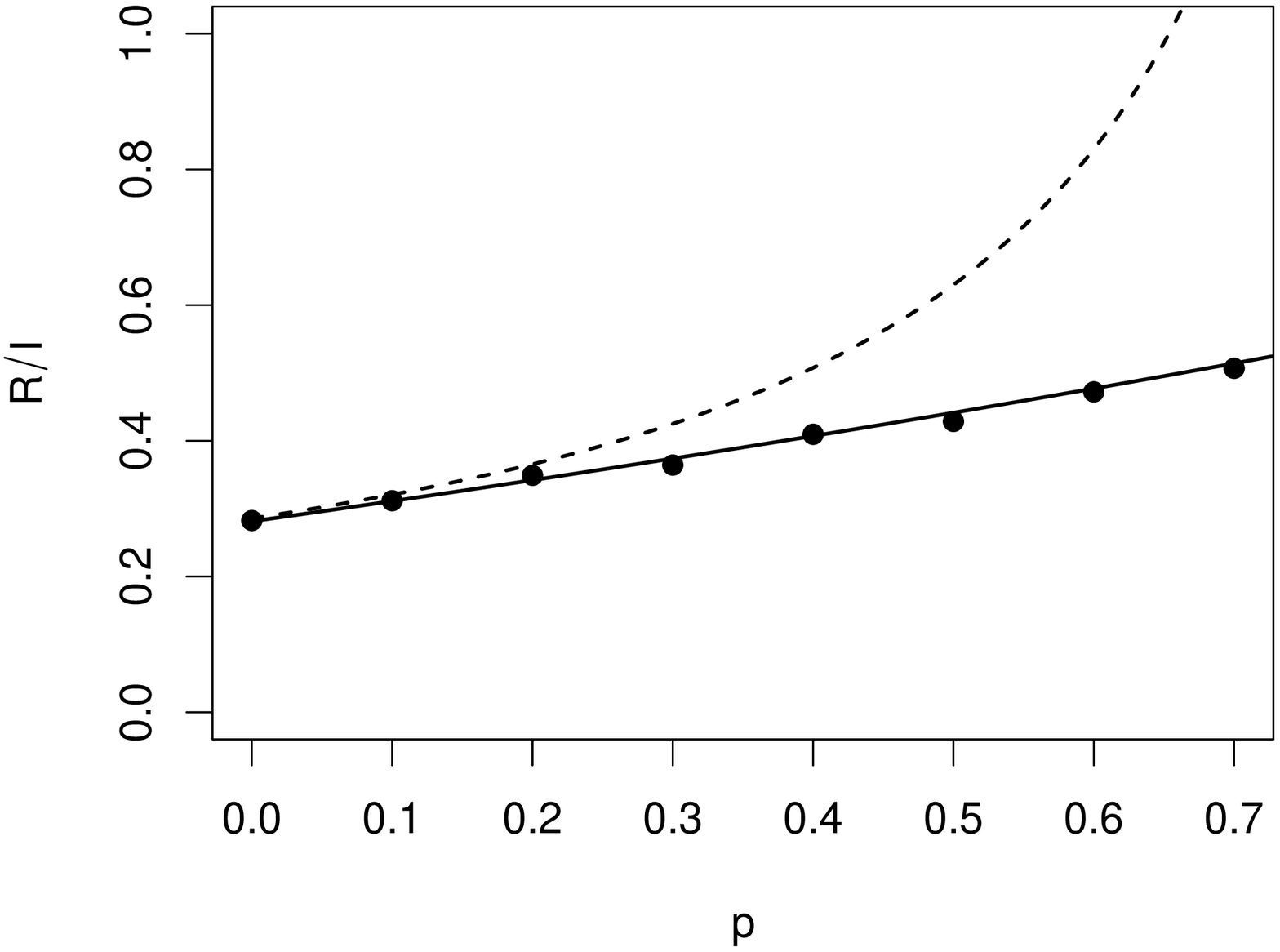}}}
  \vcenter{\hbox{\includegraphics[height=5cm]{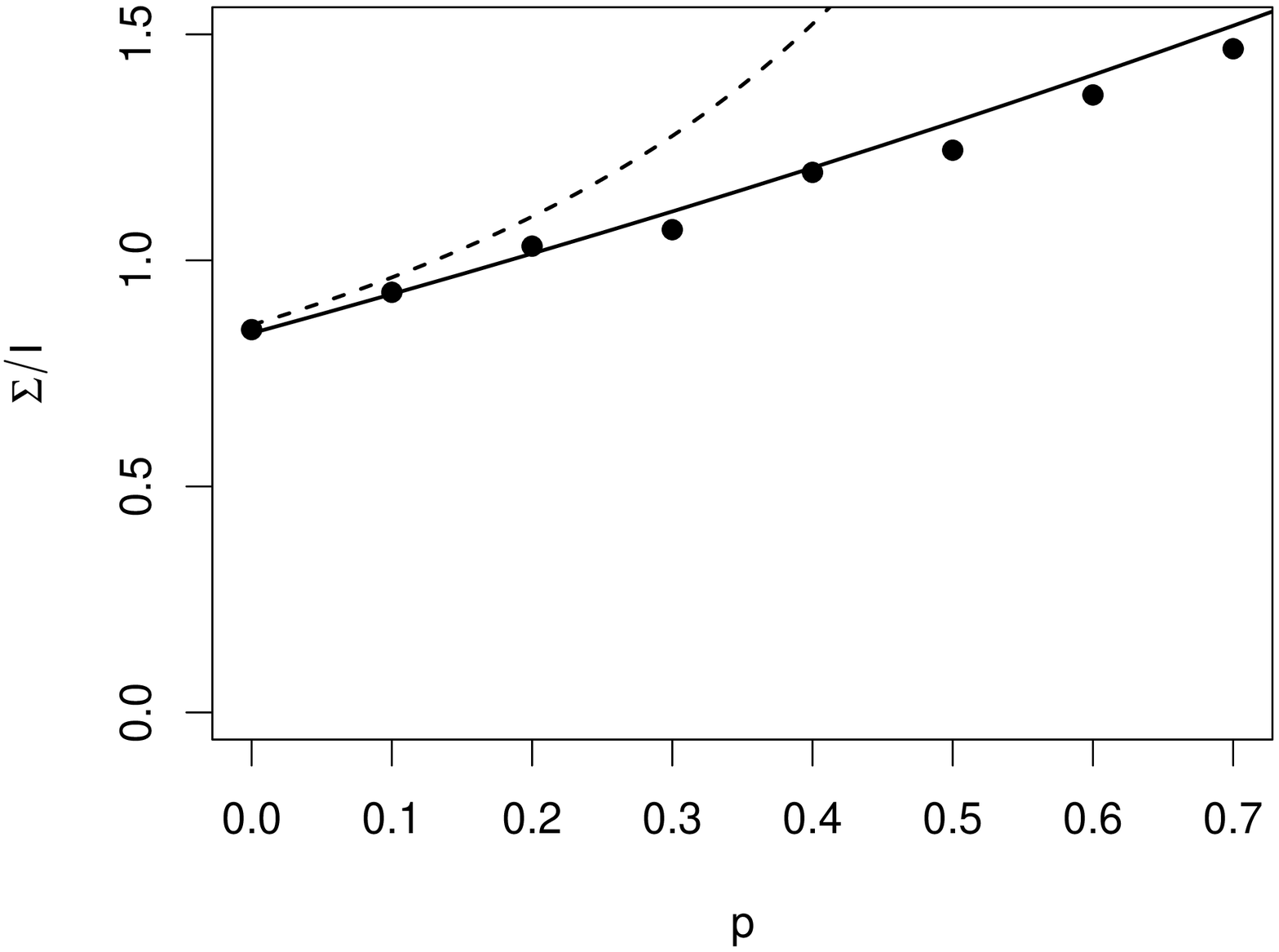}}}$
\caption{$\lambda$ (left panel), $\mathop {\lim }\limits_{t \to \infty } \frac{{R(t)}}{{I(t)}}$ (middle panel) and $\mathop {\lim }\limits_{t \to \infty } \frac{{\Sigma (t)}}{{I(t)}}$ (right panel) with different values of $p$. The bullets are the average simulating results in 50 runs, solid line is the predictions with the age structured non-constant hazard rate ($log(\kappa(a))'$) while the dashed line is the prediction with the non-age structured removal rate ($1/{\int_0^\infty  {\kappa (a)da} }$).}\label{fullOneStepExponentialGrowth}
\end{figure}


\subsection{Mean field on the configuration model}

Obviously, trees are rather mathematically convenient but no appropriate models for natural contact graphs. The configuration model (CM) is better suited ~\citep{Miller2011}. In the CM, the number of nodes $N$ and the degree distribution of nodes are given. The CM is constructed in generating for each node a random number of stubs, according to the degree distribution. Afterwards, subs are paired randomly to form edges.\\
Newman and co-workers developed a macroscopic, approximate  description of an SIR dynamics on a large configuration model, the message passing approach~\citep{karrer2010message,wilkinson2014message}. Interestingly, this approach resembles the idea used in the present work to obtain $\kappa(a)$, and is exact on trees. In the centre of that theory is the probability $H(t)$ that the infection has not been passed to a focal individual by a given edge at time $t$. Let us assume that no node is recovered at time $t=0$, and that nodes are independently and randomly assigned to be either susceptible (probability $z$)) or infected (probability $1-z$). Furthermore, $\beta$ is the rate of contacts on a given edge, and $r(a)\, da$ is the probability for removal of an infected individual at age $a$. As we lose the graph structure, we cannot speak of upstream- and downstream edges. Let $\hat K$ denote the degree distribution (all edges) of a node. 
Furthermore, $G(s)$ denotes the generating function of that degree distribution, and $G_1(s)=G'(s)/G'(1)$ is the degree distribution of a randomly chosen neighbour of a randomly chosen individual \citep{Newman2001}.
Then~\citep{karrer2010message}
\begin{eqnarray}\label{H-eqn}
H(t) = 1-\int_0^t\beta e^{-\beta a} \int_a^\infty r(\tau)\, d\tau\, (1-z G_1(H(t-a)))\, da.
\end{eqnarray}
The probabilities $P_S(t)$, ($P_I(t)$, and $P_R(t)$) for a randomly chosen individual to be susceptible (infected, recovered) at time $t$ are given by~\citep{karrer2010message}
\begin{eqnarray}
P_S(t) & = & z G(H(t))\\
\frac d{dt} P_I(t) &=&-\frac d {dt}P_S(t)-(1-z)r(t)+\int_0^t r(t-t')\frac d{dt'}P_S(t')\, dt',\qquad P_I(0)=1-z\\
P_R(t) &=& 1-P_S(t)-P_I(t).
\end{eqnarray}
For $p=0$, we are in the setting of the message passing method, where 
$$r(a)=(\alpha+\sigma)\, e^{-(\alpha+\sigma)a}.$$
And indeed, Fig.~\ref{messPassFig} (a) shows an excellent agreement of the prediction given by that method and simulations.\par\medskip

\begin{figure}[h!]
\begin{center}
(a)\includegraphics[width=7cm]{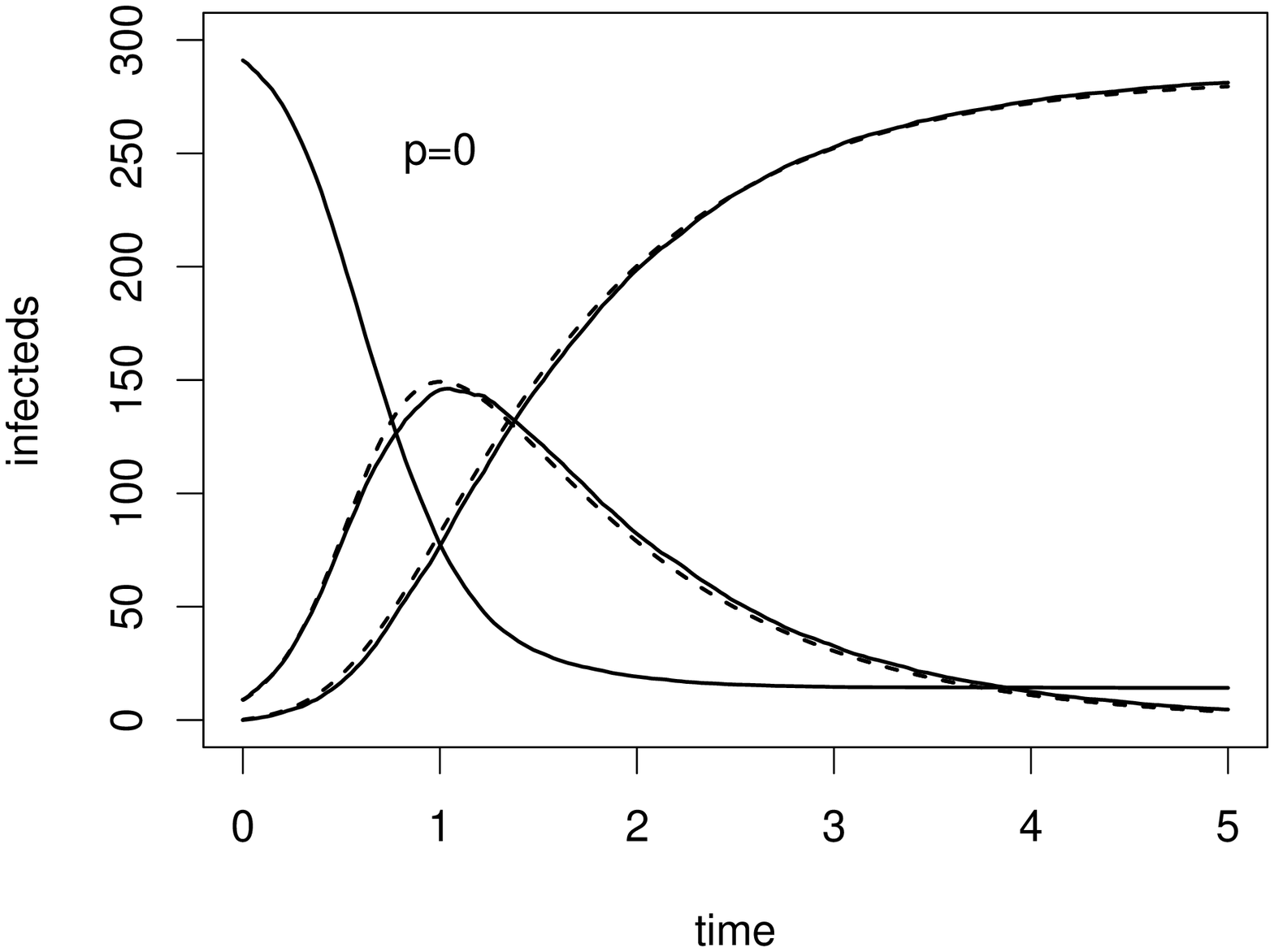}
(b)\includegraphics[width=7cm]{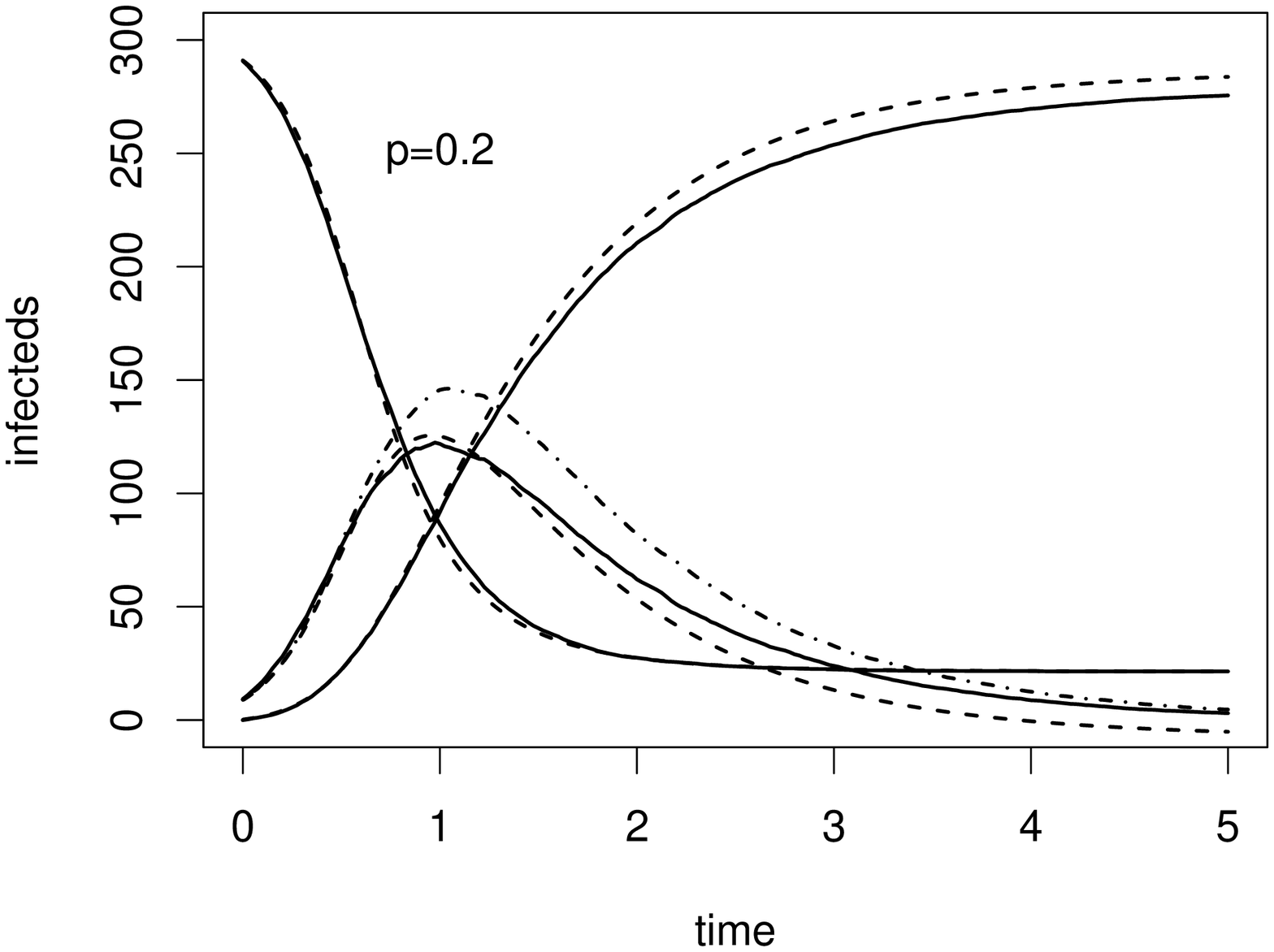}\\
(c)\includegraphics[width=7cm]{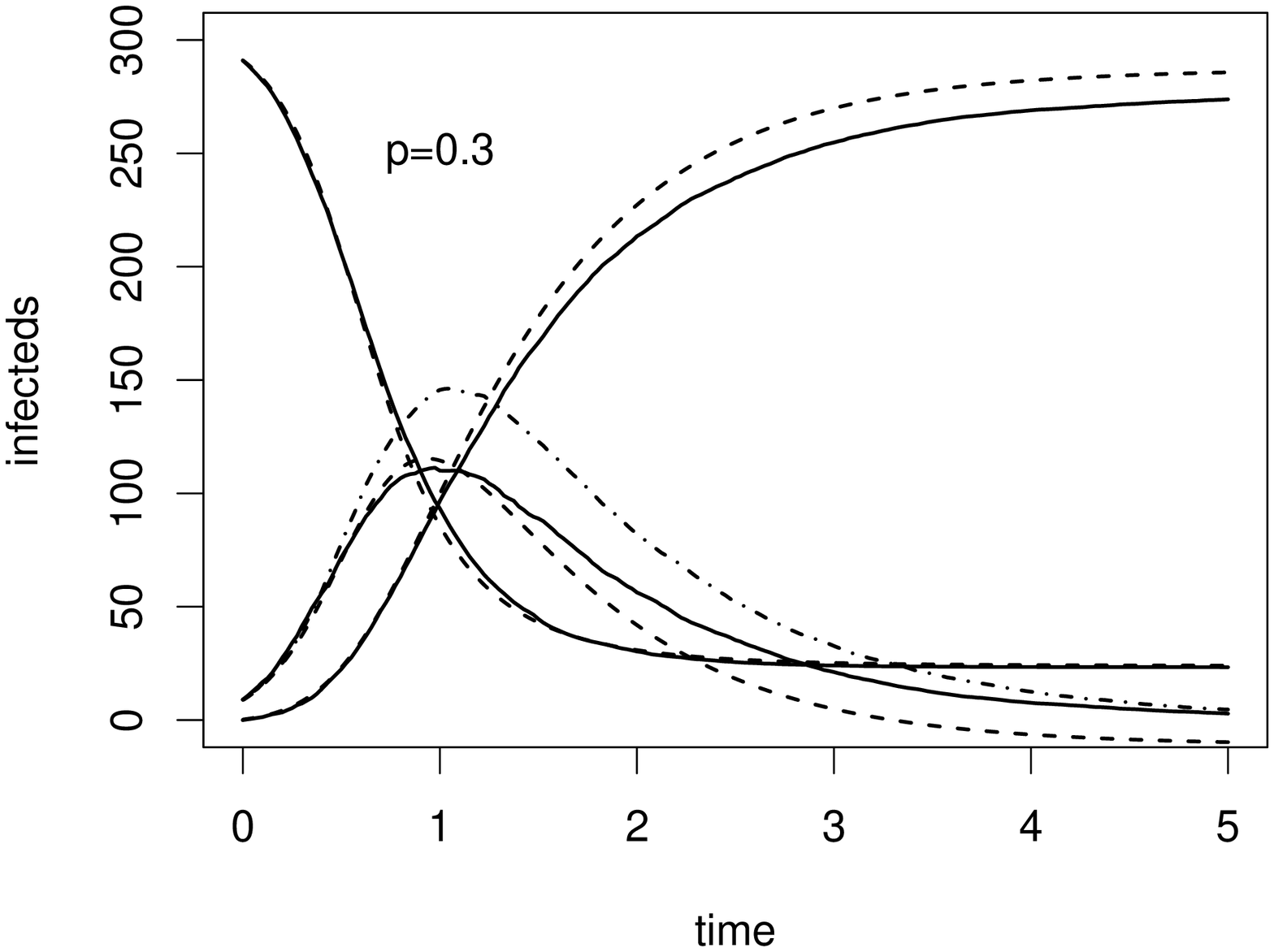}
(d)\includegraphics[width=7cm]{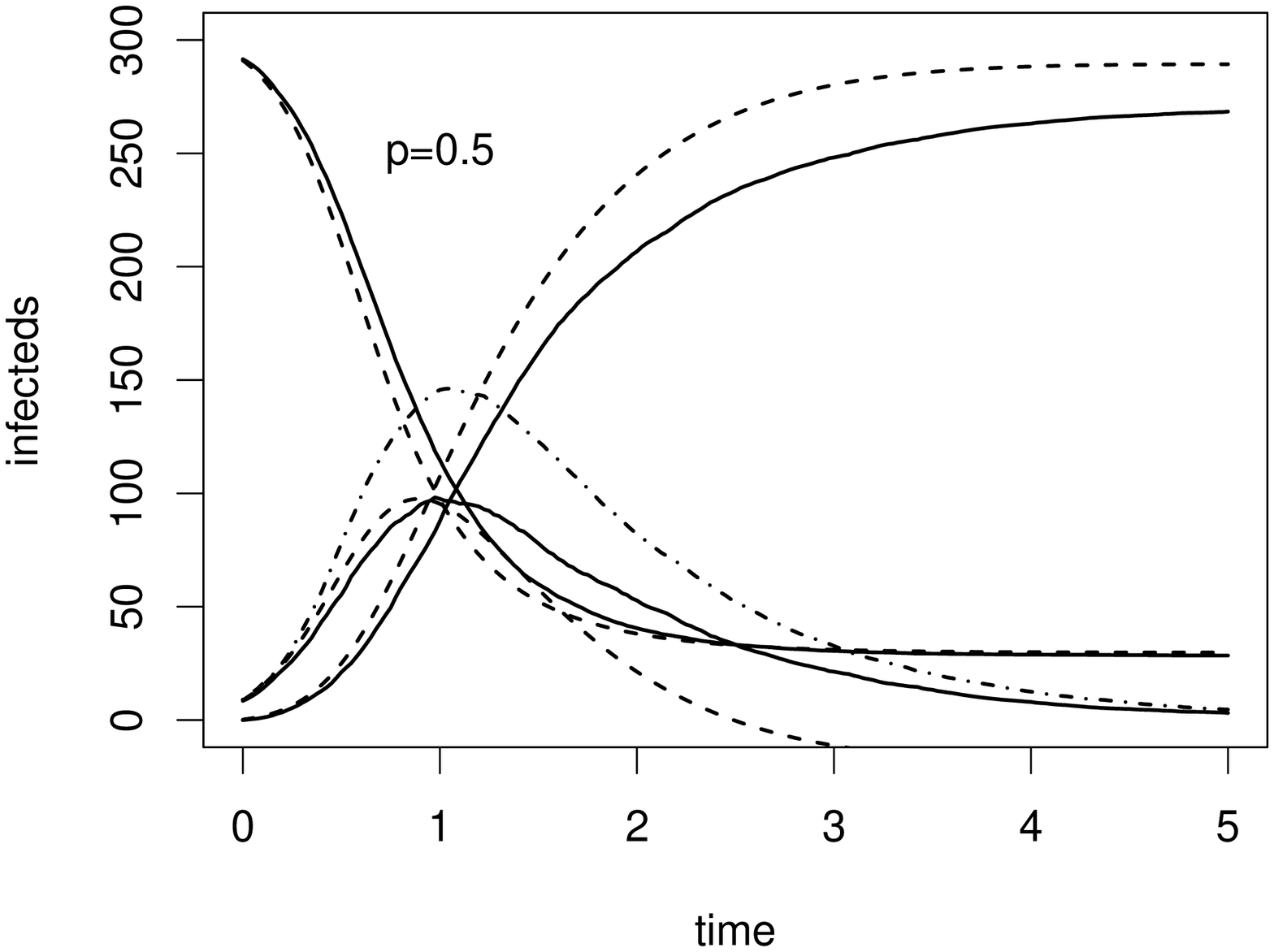}\\
\end{center}
\caption{Mean field approximation on the CM model for full tracing and different values of $p$ (as indicated in the graphs). The decreasing functions represent $S$, the increasing function $R$, and the unimodal function $I$. 
The solid lines are the average of 50 simulations, the dashed line the predictions of the message passing method for contact tracing (see text). The dashed-dotted line in panels (b)-(d) is the simulating result for $I$ and $p=0$. Parameters: $N=300$, $\beta=1.5$, $\sigma=\alpha=0.5$, Poissonian edge-distribution with $E(K)=4$. }\label{messPassFig}
\end{figure}

We incorporate contact tracing into that model in replacing $r(a)$ by $-\kappa'(a)$, where we choose the first order approximation as developed in eqn.~(\ref{fullRecApprox}) for $\kappa(a)$. We focus here on one-step tracing (see discussion below). However, there is a fundamental difference between a tree where only the root is infected and the CM: it is possible that neighbouring nodes are infected and trigger tracing events, though the focal node is neither an infectee nor infector (see also \cite{muller2000contact}). We can estimate the rate to be traced by such a contact at age of infection $a$ in multiplying the number of edges minus the edge of the infector ($E(\hat K)-1$) times the probability that no  contact happened till the present age $e^{-\beta a}$, times the probability that the contactee is infected ($P_I(t)$), times the rate of direct detection $\sigma$ times the tracing probability $p$. We obtain
\begin{eqnarray}\label{fullRecMeanfield}
\kappa(a, P_I) &:=& \hat \kappa (a)
          \bigg(1 + \frac{p\sigma (E(\hat K)-1)}{\alpha  + \sigma  - \beta} \left(e^{-\beta a} - \hat\kappa(a) \right)
           - \frac{p\sigma (E(\hat K)-1)}{\alpha  + \sigma }
          \left(1-\hat \kappa (a)\right)\\
&&\qquad\qquad\qquad  
     - \frac{p\sigma }{\alpha+\sigma}\left(1 - \hat \kappa (a)
	 \right)
 -p\,(E(\hat K)-1)\sigma P_I(t) \int_0^ae^{-\beta a'}\, da'
  \bigg).\nonumber
\end{eqnarray}
We assume that the infectious period of an individual is shorter than the time scale at which $P_I$ does change. Therewith, we replace in eqn.~(\ref{H-eqn}) $r(\tau)$ by $-\frac{\partial}{\partial\tau}\kappa(\tau,P_I(t))$. \\ 
We find for small values of $p$ a reasonable approximation (Fig.~\ref{messPassFig} (b), (c); however for $p=0.5$ and larger the residuals become notable (Fig.~\ref{messPassFig} (d)). As we only use a first-order approximation as the basis for our definition of $\kappa(a)$, we cannot expect a better result. \\
For recursive tracing, a reasonable approximation is more involving, as tracing can be triggered by individual in a distance of several steps (in the graph metric). The correlations between individuals, as well as the dependency on the background prevalence becomes more involving to control.

In principle, the message passing method is equivalent with some variation of the pair approximation~\cite{sharkey2013kj,wilkinson2014message,kiss2015generalization}. As the pair approximation is the second, frequently used approach to describe contact tracing, we have here a bridge between the present approach and the pair approximation approach to contact tracing.

\section{Discussion}

The present paper extends the existing branching-process-theory for contact tracing in a randomly mixing population to a situation where contacts are only allowed along a prescribed, non-trivial contact graph. We restricted ourselves to an SIR model, while in homogeneous populations also SIS models can be handled. The central difference between the two situations is that in homogeneous models with large populations repeated contacts between the same individuals during the infectious period can be neglected, while they do occur most likely in graph models. The analytical machinery we did use cannot handle reinfection by an infectee, and therefore we did focus on the SIR model. Herein, the mathematical technique resembles that of the message passing method, which does also 
only apply to SIR models.\par\medskip 

For the exact results in the present paper, the contact graph is chosen to be a (stochastic) tree. 
While the tree of infecteds is dynamically constructed in homogeneous models (at a given time point, the nodes are the infecteds, and a directed edge goes from infector to infectee), in the present case we focus on the tree of possible contacts, which is prescribed and static. In the present case, the tree of infecteds is a subtree of the contact tree. In order to better understand the relation between the different approaches, we clarify (some of) the time scales involved.\\ 
In particular, for real-world populations, we identify four relevant time scales: (a) The time scale at which the contact graph changes. 
A contact network is the abstraction of e.g.\ pair formation and sexual relations, family structure, or contacts at work. All these relations are likely to change on a slow time scale (children may grow up and leave the core family, a person changes his/her job position etc.). (b) The time scale at which a given edge within the contact graph is activated. Of course, there are different contacts with different intensities. The time scale for contacts within the family may be hours or minutes, while the time scale of contacts in leisure activities as a chorus may be given by a week etc. (c)~The time scale of the duration of an epidemic outbreak, and (d)~the time scale of an individuals' infectious period.\\
Only if the time scale of an infectious outbreak (c) is much shorter than that time scale describing the change in social interactions (a), we can assume the contact graph to be static. That might not be the case for endemic diseases, or infections with a long infectious period as HIV. Also asymptomatic cases of gonorrhoea and chlamydia may have an infectious period around a year, such that a static contact graph may not be suited to cover all aspects of the disease dynamics.\\
For the difference of the contact graph and the graph of infecteds, the time scale of the infectious period (d) and that of contacts on a single edge (b) becomes interesting. If the infectious period is long in comparison with the frequency of contacts on a given edge, this edge will most likely spread the infection; there is almost no difference between the contact graph and the graph of infecteds. In the contrary case, the graph of infecteds is a distinct subgraph of the contact graph. In that case, ''contact tracing'' will be rather ''infectious contact tracing''. There is no difference for the effect of contact tracing if we focus on the (rather abstract) setting of a contact graph that is a tree, where initially only the root is infected. Any downstream contactee can be only infected by an upstream individual, and the effect of contact tracing 
and infectious contact tracing coincides. That becomes different in more complex and realistic models (e.g.\ the configuration model). Downstream individuals may get infected by another downstream individual. If we focus on contact tracing, that infected downstream individual may be detected, while in infectious contact tracing, that infected downstream infected may be missed.\\
In the view of these considerations, it is an encouraging insight, that there is a limit that shows that contact tracing and infectious contact tracing lead to similar results for the appropriate scaling of contract tree and contact rates. There is no fundamental difference between homogeneous models and contact graph models, but only a gradual one.\par\medskip 

A further interesting finding of the present work is the fact that the degree distribution of the contact graph influence the dynamics mainly via the expected number of edges (at least, if the tracing probability $p$ is small, which is given in most applications). The higher moments play only a role if they become large (or even infinite). Also in the pair approximation approach, the first moment of the edge distribution enters the equations, but not the higher moments~\citep{Keeling1997}. This result may possibly contradict the idea that the effect of contract tracing relies on the detection of super-spreaders: Super-spreader strongly influence the edge distribution of a randomly chosen neighbour of a randomly chosen node; the existence of super-spreaders, however, is expressed by the second moment of the edge distribution (the variance). As this second moment is only of minor importance for the description of the overall effect of contact tracing, it might be, that also the detection of super-spreaders is not that central. This interpretation of our results is in line with the findings in~\citep{Kiss2005,Kiss2007}, where the authors showed that contact tracing works well in assortatively and associatively mixing contact graphs. Contact tracing seems to be rather robust w.r.t.\ the details of the contact graph structure. \par\medskip

In our opinion, the present work indicates that we gain a better and better insight into contact tracing on a complete graph (homogeneous model) or on a static contact network. However, the interplay of the different time scales involved (infectious period, number of contacts per edge, reformation of the contact network) is still not well understood, and requires refined models with dynamic contact graphs, and analytical methods to handle these models.

\par\bigskip
{\bf Acknowledgements}
\par\medskip
{\it This research is supported by a grant from the German Academic Exchange Service DAAD (AO), and by the Deutsche Forschungsgemeinschaft (DFG) through TUM International Graduate School of Science and Engineering (IGSSE), GSC 81, 
within the project GENOMIE\_QADOP (JM).}

\par\bigskip
{\bf Declaration of Competing Interest}
\par\medskip
{None}

\begin{appendix}
\section{A note on the simulation algorithms used}
\label{simulHints}
We simulate the SIR process (with contact tracing) on trees and on the configuration model. For the basic SIR epidemic, the Gillespie-algorithm~\citep{Doob1945,Gillespie1976} is well suited: For each process (contacts, recovery, detection), there is a clock with exponentially distributed waiting times that triggers the corresponding events. \\
In case of detection, a tracing event is triggered, and all edges are investigated, and infected nodes are removed with probability $p$. In case of trees, where only the root is the primary infected, the direction of edges and the direction in which the infection is spread coincides; in case of the configuration model, the infector/infected relation has to be stored in the nodes as an additional information. In that way, it is possible to only focus on forward or backward tracing.\\
For simulations on a tree, it is sensible to generate the tree on the fly: If an individual becomes infected, only then the (susceptible) children are generated. As the number of nodes grow exponentially with the generation, this version of the algorithm reduces the simulation time considerably.

\end{appendix}

\par\bigskip
{\bf Supplementary material}

Supplementary material associated with this article can be found, in the online version, at \url{10.1016/j.mbs.2020.108320}.
	
	\bibliographystyle{elsarticle-harv}
	\bibliography{references}


\end{document}